\documentclass[a4paper,onecolumn,11pt,unpublished]{quantumarticle}
\pdfoutput=1
\usepackage[utf8]{inputenc}
\usepackage[english]{babel}
\usepackage[T1]{fontenc}
\usepackage{amsmath}
\usepackage{amssymb}
\usepackage{mathtools}
\usepackage[colorlinks = true, citecolor= blue, urlcolor= blue, linkcolor = blue]{hyperref}
\usepackage[numbers,sort&compress]{natbib}
\usepackage[normalem]{ulem}

\usepackage{tikz}
\usepackage{lipsum}

\usepackage{amsthm}
\theoremstyle{plain}
\newtheorem{thm}{Theorem}
\newtheorem{lem}[thm]{Lemma}

\theoremstyle{definition}
\newtheorem{mydef}[thm]{Definition}

\newcommand{\A}{\mathcal{A}}
\newcommand{\B}{\mathcal{B}}
\newcommand{\X}{\mathcal{X}}
\newcommand{\Y}{\mathcal{Y}}
\newcommand{\Ah}{\hat{\mathcal{A}}}
\newcommand{\Bh}{\hat{\mathcal{B}}}
\newcommand{\Xh}{\hat{\mathcal{X}}}
\newcommand{\Yh}{\hat{\mathcal{Y}}}

\newcommand{\R}{\mathcal{R}}
\newcommand{\Q}{\mathcal{Q}}
\renewcommand{\H}{\mathcal{H}}
\newcommand{\E}{\mathcal{E}}
\newcommand{\F}{\mathcal{F}}
\renewcommand{\P}{\mathcal{P}}

\newcommand{\freq}{\mathrm{freq}}
\renewcommand{\Pr}{\mathrm{Pr}}
\newcommand{\tr}{\mathrm{tr}}
\newcommand{\vol}{\mathrm{vol}}
\newcommand{\der}{\mathrm{d}}
\newcommand{\id}{\mathrm{id}}

\newcommand{\norm}[1]{\left\lVert#1\right\rVert}
\newcommand{\simplex}[1]{\Delta^{#1}}
\newcommand{\cl}[1]{\operatorname{cl}(#1)} 

\newcommand{\Reals}{\mathbb{R}}

\newcommand{\Integers}{\mathbb{N}}

\newcommand{\edit}[2]{#2}
\newcommand{\editB}[2]{#2}

\renewcommand{\[}{
\begin{equation}
}
\renewcommand{\]}{\end{equation}}

\begin{document}

\title{De Finetti Theorems for Quantum Conditional Probability Distributions with Symmetry}

\author{Sven Jandura}
\affiliation{Institut de Science et d’Ingénierie Supramoléculaires (UMR 7006), University of Strasbourg, 8 Allée Gaspard Monge, 67000 Strasbourg, France}
\affiliation{Institute for Theoretical Physics, ETH Zurich, 8093 Zurich, Switzerland}
\orcid{0000-0003-0282-7637}
\author{Ernest Y.-Z. Tan}
\affiliation{Institute for Quantum Computing and Department
of Physics and Astronomy, University of Waterloo, Waterloo, Ontario N2L 3G1, Canada}
\affiliation{Institute for Theoretical Physics, ETH Zurich, 8093 Zurich, Switzerland}
\orcid{0000-0003-4872-158X}

\begin{abstract}
The aim of device-independent quantum key distribution (DIQKD) is to study protocols that allow the generation of a secret shared key between two parties under minimal assumptions on the devices that produce the key. These devices are merely modeled as black boxes and mathematically described as conditional probability distributions. A major obstacle in the analysis of DIQKD protocols is the huge space of possible black box behaviors. De Finetti theorems can help to overcome this problem by reducing the analysis to black boxes that have an iid structure. Here we show two new de Finetti theorems that relate conditional probability distributions in the quantum set to de Finetti distributions (convex combinations of iid distributions), that are themselves in the quantum set. We also show how one of these de Finetti theorems can be used to enforce some restrictions onto the attacker of a DIQKD protocol. Finally we observe that some desirable strengthenings of this restriction, for instance to collective attacks only, are not straightforwardly possible.
\end{abstract}

\section{Introduction}
The aim of quantum key distribution is to establish a shared key between two parties, commonly called Alice and Bob, that is unknown to any third party, commonly called Eve. To achieve this goal, Alice and Bob can share an entangled quantum state and use the correlated outcomes of measurements on this state to generate a secure key pair via a postprocessing protocol. If Eve has tampered with the shared state, Alice and Bob either notice this and abort the protocol, or are able to generate a secure key pair anyway \cite{Bennett2014, Ekert1991, Renner2008}. In device-independent quantum key distribution (DIQKD) we assume that Eve not only has control over the shared state, but is also able to manipulate the devices that Alice and Bob use to measure the state. As long as the devices are not manipulated in a way that sends information out of Alice's and  Bob's laboratories through channels other than those controlled by Alice and Bob, there are still protocols that allow for the generation of a shared secret key \cite{Barrett2005, Acin2007, Pironio2009, Vazirani2014, ArnonFriedman2018}.

In the device-independent context, the devices of Alice and Bob are treated as black boxes and modeled by a conditional probability distribution $P_{AB|XY}$. Alice and Bob can give inputs $x$ and $y$ respectively to the box, and receive outputs $a$ and $b$ with probability $P_{AB|XY}(ab|xy)$. We will often write $P(ab|xy)$ instead of $P_{AB|XY}(ab|xy)$ when the random variables $A$, $B$, $X$ and $Y$ are implicitly understood.  In the device dependent case the inputs correspond\edit{s}{} to the choice of measurement basis, and the outputs to the results of the measurement. We will denote the sets of possible inputs by $\X$ and $\Y$, and the sets of outputs by $\A$ and $\B$. If the inputs and outputs are strings of length $n$, i.e.~they are of the form $\X = \Xh^n$ \edit{}{(where $\Xh$ denotes some set of possible single round inputs)} and analogously for $\Y$, $\A$ and $\B$, we call $P_{AB|XY}$ an $n$-round box. If $Q_{\hat{A}\hat{B}|\hat{X}\hat{Y}}$ is a box with inputs and outputs in $\Xh,\Yh,\Ah,\Bh$ we denote by $Q^{\otimes n}_{AB|XY}$ the $n$-round iid box with
\[
Q^{\otimes n}(ab|xy) = \prod_{i=1}^n Q(a_ib_i|x_iy_i) .
\]

Not all boxes $P_{AB|XY}$ describe processes that are physically possible if we assume that no information can leave the laboratories of Alice and Bob. All boxes must then be such that Bob gains no information about Alice's input from his output, and vice versa. We refer to boxes satisfying this constraint as non-signaling:
\begin{mydef}
A box $P_{AB|XY}$ is non-signaling if
\begin{align}
\forall b,y,x,x' \qquad \sum_a P(ab|xy) &= \sum_a P(ab|x'y)\\
\forall a,x,y,y' \qquad \sum_b P(ab|xy) &= \sum_b P(ab|xy').
\end{align}
\end{mydef}
If we furthermore assume that the boxes are described by quantum theory, we can describe the distribution $P_{AB|XY}$ by some quantum state shared between Alice and Bob and some POVMs describing their measurements. 
\begin{mydef}
A box $P_{AB|XY}$ is quantum if there are Hilbert spaces $\H_A$, $\H_B$, a state $\rho_{AB} \in \mathrm{End}\left(\H_A \otimes \H_B\right)$, for each $x$ a POVM $\left\lbrace E^{a,x} | a \right\rbrace$ \edit{}{on $\H_A$} and for each $y$ a POVM $\left\lbrace F^{b,y} | b \right\rbrace$ \edit{}{on $H_B$} such that
\[
P(ab|xy) = \tr\left[\rho_{AB} \left(E^{a,x} \otimes F^{b,y}\right)\right].
\]
\end{mydef}
\noindent The set of quantum boxes is a proper subset of the set of non-signaling boxes. That both sets are not identical is demonstrated by the Popescue-Rohrlich box \cite{Popescu1994}.  

When constructing DIQKD security proofs, often the analysis would be simplified if there were some form of reduction from general box behaviour to the iid case, as it is substantially easier to construct security proofs for the latter (as achieved in e.g.~\cite{Acin2007, Pironio2009, Tan2020}). To find such a reduction, so-called de Finetti theorems may be a promising tool, as they have previously been used to achieve this goal in the case of device-dependent QKD~\cite{Renner2008,Christandl2009a}.  De Finetti theorems allow us to relate the entries of an arbitrary permutation invariant box to the entries of a de Finetti box (a convex combination of iid boxes). De Finetti theorems where originally developed for random variables \cite{Diaconis1980} and then extended to quantum states \cite{Hudson1976, Caves2002, Christandl2007, Renner2008} and boxes \cite{ArnonFriedman2015, Christandl2009}. 
For example, in \cite{ArnonFriedman2015} it was shown that for each set of \editB{}{single round inputs $\Ah$ and outputs $\Xh$} there exists a de Finetti box $\tau_{A|X}$ such that for all permutation invariant boxes $P_{A|X}$ it holds that
\[
\forall a\in \A,\, x\in\X \qquad P(a|x) \leq (n+1)^{\editB{}{|\Xh|(|\Ah|}-1)} \tau(a|x).
\label{eq:de_finetti_theorem_general}
\]
(Here we treat the inputs and outputs of Alice and Bob as lumped together to a single input and output.)

\editB{}{
However, the de Finetti theorems for boxes derived in e.g.~\cite{ArnonFriedman2015, Christandl2009} have the drawback that the de Finetti boxes cannot be restricted to the quantum set even if the original permutation invariant boxes are quantum. This creates an obstacle for applications, because many existing DIQKD security proofs under the iid assumption exploit the properties of the quantum set \cite{Acin2007, Pironio2009, Tan2020}. This implies that such proofs cannot be combined with the de Finetti theorems in~\cite{ArnonFriedman2015, Christandl2009} to obtain security against non-iid attacks, as those de Finetti theorems involve boxes that are not in the quantum set. (It is true that one could aim to derive a security proof for all iid behaviours in the non-signaling rather than quantum set, then apply the de Finetti theorems of~\cite{ArnonFriedman2015, Christandl2009} to obtain security against non-iid attacks. However, this would give lower asymptotic keyrates and noise tolerance compared to security proofs against quantum attackers, because non-signaling behaviours yield a significantly larger class of possible attacks.)
Ideally, we would like to find a de Finetti theorem that can extend the iid security proofs against quantum attackers in~\cite{Acin2007, Pironio2009, Tan2020} to cover non-iid quantum attackers, while preserving the asymptotic keyrates and noise tolerance from those proofs, similar to the situation for device-dependent QKD~\cite{Christandl2009a}. While we do not fully achieve this goal in this work, we do obtain a de Finetti theorem that allows a partial reduction to the iid case (in a sense described in section~\ref{sec:applications}), and we also highlight some concrete difficulties that may be faced when aiming for a full reduction.

Regarding other existing approaches for reductions to the iid case, we note that for DIQKD protocols that use only one-way communication for error correction \cite{Renner2008}, a proof technique known as the entropy accumulation theorem (EAT) \cite{Dupuis2020} can be used to essentially reduce the analysis of non-iid (but time-ordered) boxes to the iid scenario \cite{ArnonFriedman2018}. Alternatively, the techniques in~\cite{Jain2020,Vidick2017} can be used to obtain security proofs for such protocols even when the boxes accept all inputs in parallel, though the resulting asymptotic keyrates are lower than in the iid case. There are however protocols that don't only use one-way error correction (broadly referred to as advantage distillation protocols \cite{Maurer1993,Wolf1999,Chau02,Gottesman2003,Renner2008,Tan2020}, such as the Cascade protocol \cite{Brassard1994} or the repetition-code protocol \cite{Maurer1993,Renner2008}), and these protocols do not admit a security proof via those approaches.\footnote{\editB{}{More specifically: the proof approaches in those works essentially rely on bounding the conditional smooth min-entropy~\cite{Renner2008} of the ``raw'' box outputs (or a subset thereof), then compensating for the additional information revealed during one-way error correction by subtracting the number of bits communicated during that step. However, advantage distillation protocols may perform a significant amount of processing on the box outputs \emph{before} the privacy amplification step (informally: the step in which the data is transformed into the final key), and furthermore they often communicate a very large number of bits in the process as compared to one-way error correction (see e.g.~the repetition-code protocol~\cite{Maurer1993,Renner2008}). Hence the proof techniques in~\cite{ArnonFriedman2018,Jain2020,Vidick2017} are difficult to extend to these advantage distillation protocols.}}
The significance of these protocols in DIQKD is that under an iid assumption, it has been shown~\cite{Tan2020} that they can achieve higher noise tolerances than one-way error correction (i.e.~they can achieve positive keyrates even when the keyrate given by one-way error correction is zero), analogous to results for device-dependent QKD~\cite{Chau02,Gottesman2003,Renner2008}. 
However, for device-dependent QKD these improved noise tolerances can be lifted to the non-iid case using de Finetti arguments as mentioned above, whereas in DIQKD such an argument is currently missing --- in fact, there are currently no security proofs for DIQKD advantage distillation protocols against non-iid attacks. Finding a way to resolve this would be useful in, for instance, tackling a foundational question of characterizing which nonlocal box behaviours can be used for DIQKD~\cite{Farkas2021} (analogous to the question of \textit{bound information} in device-dependent QKD~\cite{Gisin2000}), since advantage distillation can have higher noise tolerances than one-way error correction.
}

Our main result in this work consists of two de Finetti theorems for Clauser-Horne-Shimony-Holt (CHSH) symmetric quantum boxes (see definition~\ref{def:chsh_symmetry}), such that the de Finetti box is quantum as well. We further show how the first de Finetti theorem could be used in the security proofs of DIQKD protocols, \editB{}{yielding a partial reduction to the iid case.}

The rest of this paper is structured as follows: In section~\ref{subsec:first_de_finetti_theorem} we show the first de Finetti theorem (theorem \ref{thm:de_finetti_chsh_symmetric}). It is similar to eq.~\eqref{eq:de_finetti_theorem_general} and shows that the entries of a CHSH symmetric quantum box are upper bounded, up to a factor polynomial in $n$, by the entries of a fixed quantum de Finetti box. In section~\ref{subsec:second_de_finetti_theorem} we then show the second de Finetti theorem (theorem \ref{thm:de_finetti_theorem_2}), which is closer to the original de Finetti theorems for random variables and quantum states. It states that the marginal of the first $k$ rounds of a $n$-round CHSH symmetric quantum box is close to (and not just upper bounded by) a quantum de Finetti box. Our results in this section rely on the existence of appropriate threshold theorems (see e.g.~theorem~\ref{thm:threshold_unger} below). A natural question is whether it is possible to derive them without using the threshold theorems; however, we show in appendix~\ref{sec:a_de_Finetti_theorem_for_general_symmetries} that proving a de Finetti theorem of the first form is essentially equivalent to proving a threshold theorem, hence it would be a result of comparable difficulty. 

In light of this, our results cannot currently be used as an alternative method to prove threshold theorems. However, our focus is more on the application of these results for DIQKD security proofs.
Hence in section~\ref{sec:applications}, we present an application of the first de Finetti theorem to bound the diamond distance between two channels acting on boxes. 
The diamond distance measures how well these two channels can be distinguished by an attacker. Since the security of a DIQKD protocol is related to the diamond distance between the protocol and an ideal channel \cite{Maurer2011, Portmann2021}, bounds on the diamond distance can be useful in DIQKD security proofs. We show that to prove security of a DIQKD protocol against arbitrary (so-called coherent \cite{Scarani2012}) quantum attacks it is sufficient to prove security against an adversary who holds an extension of a fixed quantum de Finetti box (theorem~\ref{thm:postselection_ns}). However, this extension may not be quantum itself and can only be restricted to the non-signaling set. 

In section~\ref{sec:impossible_strengthening} we  show that the result from section~\ref{sec:applications} cannot be strengthened to restrict the attacker further to collective attacks \cite{Scarani2012} (attacks where the black box can be described by an iid quantum state and iid measurements for Alice and Bob, see definition~\ref{def:collective_attack_box}). For this, we construct two channels that cannot be distinguished at all using boxes compatible with collective attacks, but can be distinguished if arbitrary quantum boxes are available (theorem~\ref{thm:no_postselection_for_collective_attacks}). This shows that the theorem from section~\ref{sec:applications} cannot be immediately used to conclude security against coherent attacks from security against collective attacks.

\section{De Finetti Theorems for boxes with CHSH symmetry}
\label{sec:de_finetti_theorems}
\subsection{The first de Finetti theorem}
\label{subsec:first_de_finetti_theorem}
Let us first define de Finetti boxes and CHSH symmetry:
\begin{mydef}
A $n$-round box $\tau_{AB|XY}$ is called de Finetti if it is the convex combination of iid boxes.
\end{mydef} 
 \begin{mydef}
\label{def:chsh_symmetry}
An $n$-round box $P_{AB|XY}$ with single round inputs and outputs in \editB{}{$\Ah=\Bh=\Xh=\Yh=\{0,1\}$} is called CHSH symmetric if
\[
P(ab|xy) = P(a'b'|x'y') \qquad \text{ whenever }\norm{a\oplus b \oplus xy}_0 = \norm{a'\oplus b' \oplus x'y'}_0.
\]
Here $\norm{x}_0$ denotes the number of non-zero entries of a $n$-bit string $x$.
\end{mydef}
If for an index $i$ we have $a_i\oplus b_i = x_iy_i$ we say that the CHSH game is won in round $i$ \cite{Clauser1969}. Thus, definition~\ref{def:chsh_symmetry} basically states that a box is CHSH symmetric if its entries $P(ab|xy)$ only depends on how many indices the CHSH game was won. Our definition of CHSH symmetry differs slightly from the one in \cite{ArnonFriedman2015}, where it is only required that $P(ab|xy)=P(a'b'|x'y')$ whenever $a\oplus b \oplus xy = a'\oplus b' \oplus x'y'$. Our definition agrees with the definition in \cite{ArnonFriedman2015} for permutation symmetric boxes --- essentially, we have implicitly incorporated the constraint of permutation symmetry into definition~\ref{def:chsh_symmetry} itself.

Of course an attacker can initially manipulate the boxes of Alice and Bob such that they do not possess CHSH symmetry. However, Alice and Bob can run the following procedure to enforce CHSH symmetry: First Alice chooses a random permutation $\pi$ and transmits it to Bob over the authenticated channel, then Alice and Bob permute their inputs and outputs according to $\pi$. If we view the original box as a conditional probability distribution $P_{ABE|XYZ}$ for Alice, Bob and Eve we can view $\pi$ as additional knowledge \edit{}{$E'$} of Eve and describe the box after $\pi$ has been applied by $\tilde{P}_{ABEE'|XYZ}$. \edit{}{Alice and Bob will not require $\pi$ for the remainder of the protocol and can now discard it; therefore, in the rest of our discussion we do not include it in the marginal of the Alice-Bob boxes, whereas on Eve's component we will simply absorb $E'$ into $E$ and no longer explicitly denote it.} Then the marginal $\tilde{P}_{AB|XY}$ has permutation symmetry. To go from permutation symmetry to CHSH symmetry Alice and Bob apply the depolarization protocol described in appendix A of \cite{Masanes2006} to each round. If the box in the honest implementation of the DIQKD protocol has CHSH symmetry it is unchanged by this depolarization protocol. It is important to note that only the marginal box of Alice and Bob has CHSH symmetry after this protocol: from the perspective of Eve, who knows the permutation $\pi$ and the random bits chosen in the depolarization protocol in \cite{Masanes2006}, the box may not have CHSH symmetry. \edit{}{However, we highlight that in the case of device-dependent QKD, this did not prevent constructing a security proof via de Finetti arguments~\cite{Christandl2009a}, and hence there still remains the possibility of a similar result for DIQKD.}

It was shown in \cite{ArnonFriedman2015} that a de Finetti theorem holds for CHSH symmetric boxes:
\begin{thm}[Corollary 6 in \cite{ArnonFriedman2015}]
For each number $n$ of rounds there is an $n$-round CHSH symmetric de Finetti box $\tau_{AB|XY}$ such that for all CHSH symmetric boxes $P_{AB|XY}$ it holds that
\[
P(ab|xy) \leq (n+1)\tau(ab|xy).
\]
\label{thm:de_finetti_chsh_arf}
\end{thm}
\noindent Theorem \ref{thm:de_finetti_chsh_arf} was derived for all CHSH symmetric boxes $P_{AB|XY}$, even if they are not quantum. However, the de Finetti box $\tau_{AB|XY}$ constructed in the theorem is also not quantum. 
The main result we derive in this section is a de Finetti theorem for \emph{quantum} CHSH symmetric boxes, hence resolving this issue:
\begin{thm}
\label{thm:de_finetti_chsh_symmetric}
For each number $n$ of rounds there is an $n$-round CHSH symmetric  quantum de Finetti box $\tau_{AB|XY}$ such that for all CHSH symmetric quantum boxes $P_{AB|XY}$ it holds that
\[
P(ab|xy) \leq (n+1)^2 \tau(ab|xy).
\]
\end{thm}

The maximal probability with which any single round quantum box can win the CHSH game is $w = \frac{2+\sqrt{2}}{4}$ \cite{Cirelson1980}. This value is called the quantum value of the CHSH game. To prove theorem \ref{thm:de_finetti_chsh_symmetric} we need the following specialization of a theorem from \cite{Unger2009} to the CHSH case.   It says that the probability that the fraction of won CHSH games is larger than a certain threshold (namely the value of the  CHSH game) is exponentially small in $n$. Such theorems are commonly referred to as \emph{threshold theorems}.\footnote{To be precise, theorem~\ref{thm:threshold_unger} is a ``perfect'' threshold theorem, in that the exponent in the bound~\eqref{eq:threshold_unger} is such that the bound is, up to a factor polynomial in $n$, equal to $\binom{n}{k}w^k(1-w)^{n-k}$, the probability to win exactly $k$ games if a single game is won with probability $w$. ``Imperfect'' threshold theorems can be roughly described as giving bounds of the more general form $e^{-n\Delta(k/n)}$, where $\Delta$ is some potentially ``looser'' way to quantify the distance from $k/n$ to $w$ \cite{ArnonFriedman2016}.
Our first de Finetti theorem (theorem~\ref{thm:de_finetti_chsh_symmetric}) and its generalization (theorem~\ref{thm:deFinetti_main}) both require a perfect threshold theorem. However, the proof of our second de Finetti theorem (theorem~\ref{thm:de_finetti_theorem_2}) still holds with an imperfect threshold theorem, though the resulting bound would be weaker.
}
\begin{thm}[Theorem 5 in \cite{Unger2009}]
\label{thm:threshold_unger}
Let $P_{AB|XY}$ be an $n$-round box with single round inputs and outputs in $\{0,1\}$. Let $\mu$ be the uniform probability distribution on $\{0,1\}^2$ and $K = ||A\oplus B \oplus XY \oplus \mathbf{1}||_0$ the number of won instances of the CHSH game. Let $w = \frac{2+\sqrt{2}}{4}$ the quantum value of the CHSH game. Then for $k > wn$
\[
\Pr_{P_{AB|XY},\mu^{\otimes n}}[K \geq k] \leq e^{-nD(k/n,1-k/n||w,1-w)},
\label{eq:threshold_unger}
\]
where $\Pr_{P_{AB|XY},\mu^{\otimes n}}$ denotes the probability measure in which $X$ and $Y$ are sampled from $\mu^{\otimes n}$ and $A$ and $B$ are sampled using $P_{AB|XY}$ and
\[
D(p,1-p||q,1-q) = p \ln\left(\frac{p}{q}\right) + (1-p) \ln\left(\frac{1-p}{1-q}\right)
\]
denotes the relative entropy.
\end{thm}

Note that the box $P_{AB|XY}$ in theorem \ref{thm:threshold_unger} does not have to be CHSH symmetric. However, if $P_{AB|XY}$ is CHSH symmetric then we can describe it completely by $n+1$ parameters $\{p_0,p_1,\dots,p_n\}$, which we define as follows: for each $k \in \{0,...,n\}$, take any $a,b,x,y \in \{0,1\}^n$ such that $k = ||a\oplus b \oplus xy \oplus \mathbf{1}||_0$ (in other words, $a,b,x,y$ win exactly $k$ instances of the CHSH game). Then define
\[
p_k = \binom{n}{k}2^n P(ab|xy).
\label{eq:prob_of_winning_k_games}
\]
By CHSH symmetry, all combinations of $a,b,x,y$ with the same value of $k$ have the same value of $P(ab|xy)$, so the expression~\eqref{eq:prob_of_winning_k_games} is indeed well-defined.
The normalization factor  $\binom{n}{k}2^n$ is chosen to give these parameters a simple interpretation: namely, $p_k$ is in fact equal to the probability of winning exactly $k$ CHSH games for the box $P(ab|xy)$ (regardless of the input distribution). 
To see this, notice that for fixed $x$, $y$ and $k$ there are exactly $\binom{n}{k}2^n$ pairs $a,b$ such that $k = ||a\oplus b \oplus xy \oplus 1||_0$ \cite{ArnonFriedman2015}. Therefore for any probability measure $\mu$ on the $n$-round inputs $x$ and $y$, we indeed have
\begin{align}
\Pr_{P_{AB|XY},\mu}[K=k] &= \sum_{\substack{a,b,x,y \\k = ||a\oplus b \oplus xy \oplus \mathbf{1}||_0}} P(ab|xy)\mu(xy) \nonumber \\
&=   \sum_{xy} \binom{n}{k}2^n  \frac{p_k}{\binom{n}{k}2^n} \mu(xy)\\  \nonumber
&= p_k
\end{align}
where in the second equality we used that the summand does not depend on $a$ and $b$, and for a fixed $x$ and $y$ there are $\binom{n}{k}2^n$ possible $a$ and $b$ with $\norm{a\oplus b\oplus xy\oplus \mathbf{1}}_0=k$.
Theorem \ref{thm:threshold_unger} then implies
\[
\sum_{l \geq k} p_l \leq e^{-nD(k/n,1-k/n||w,1-w)}.
\]

To prove theorem \ref{thm:de_finetti_chsh_symmetric} we need one further ingredient:
\begin{lem}
\label{lem:chsh_concavity_argument}
Let $a,b \in \Reals$ with $a < b$ and $f:[a,b]\rightarrow\Reals^+_0$ be a concave function that attains its maximum at some $x^* \in [a,b]$. Then $ \forall n \in\Integers$,
\[
\frac{1}{n+1}(b-a)f(x^*)^n \leq \int_a^b f(x)^n\der x \leq (b-a)f(x^*)^n
\]
\end{lem}
The proof is given in appendix~\ref{sec:proof_concavity_lemma}. Now we are ready to prove theorem \ref{thm:de_finetti_chsh_symmetric}:
\begin{proof}[Proof of theorem \ref{thm:de_finetti_chsh_symmetric}]
Let $w = \frac{2+\sqrt{2}}{4}$ be the quantum value of the CHSH game and define for $p\in[1-w,w]$ the single round box $Q(p)_{\hat{A}\hat{B}|\hat{X}\hat{Y}}$ as
\[
Q(p)(ab|xy) = \begin{cases} p/2 &\text{ if } a\oplus b = xy \\ (1-p)/2 &\text{ if } a\oplus b \neq xy \end{cases}
\]
Now set
\[
\tau_{AB|XY} = \frac{1}{2w-1}\int_{1-w}^w Q(p)_{AB|XY}^{\otimes n} \der p.
\]
$\tau_{AB|XY}$ is quantum because each $p \in [1-w,w]$ $Q(p)_{AB|XY}^{\otimes n}$ is quantum. Further $\tau_{AB|XY}$  is de Finetti by construction. Let $a,b,x,y \in \{0,1\}^n$ and let $k=\norm{a\oplus b \oplus xy \oplus \mathbf{1}}_0$ be the number of won CHSH games of $a,b,x,y$. Let $\alpha = k/n$ and set
\[
f(p) = \frac{1}{2} p^\alpha(1-p)^{1-\alpha}.
\]
Then 
\[
\tau(ab|xy) = \frac{1}{2w-1} \int_{1-w}^w 2^{-n}p^k(1-p)^{n-k} \der p = \frac{1}{2w-1}\int_{1-w}^w f(p)^n\der p.
\label{eq:chsh_tau_formula}
\]
Note that
\[
f'(p) = \left(\frac{\alpha}{p}-\frac{1-\alpha}{1-p}\right)f(p)
\]
and
\begin{align}
f''(p) &= \left[\left(\frac{\alpha}{p}-\frac{1-\alpha}{1-p}\right)^2-\frac{\alpha}{p^2}-\frac{1-\alpha}{(1-p)^2}\right]f(p) \nonumber \\
&= -\frac{\alpha(1-\alpha)}{p^2(1-p)^2} f(p) < 0.
\end{align}
Therefore, $f$ is concave and its maximum on the interval $[0,1]$ occurs at $p=\alpha$.

Since $f$ is concave, we can apply lemma \ref{lem:chsh_concavity_argument} to eq.~\eqref{eq:chsh_tau_formula} and get
\[
\tau(ab|xy) \geq \frac{1}{n+1}\sup_{p\in[1-w,w]} f(p)^n.
\label{eq:chsh_tau_bound}
\]

Now we turn to the box $P_{AB|XY}$. Following the earlier notation, let $p_k$ denote the probability of winning exactly $k$ CHSH games with this distribution. We observe that
\begin{itemize}
\item If $\alpha > w$ then, the threshold theorem \ref{thm:threshold_unger} implies
\[
p_k \leq e^{-nD(\alpha,1-\alpha||w,1-w)}. 
\label{eq:p_bound1}
\]
\item If $\alpha < 1-w$ we can use the threshold theorem to get a bound on the minimal number of won games, because the CHSH game has the property that winning exactly $k$ games is just as hard as losing exactly $k$ games (and thus winning $n-k$ games). Hence we have
\[
p_k \leq e^{-nD(\alpha,1-\alpha||1-w,w)}.
\label{eq:p_bound2}
\]
\item If $\alpha \in [1-w,w]$ we can rewrite the trivial bound $p_k \leq 1$ in the form
\[
p_k \leq 1 = e^{-nD(\alpha,1-\alpha||\alpha,1-\alpha)}.
\label{eq:p_bound3}
\]
\end{itemize}
Hence we can summarize the implications of the threshold theorem as
\[
p_k \leq \sup_{p \in [1-w,w]} e^{-nD(\alpha,1-\alpha||p,1-p)}.
\label{eq:thresh_summary}
\]
We can simplify the term in the supremum by inserting the definition of relative entropy:
\[
e^{-nD(\alpha,1-\alpha||p,1-p)} = \left(\frac{p}{\alpha}\right)^{\alpha n} \left(\frac{1-p}{1-\alpha}\right)^{(1-\alpha) n} = \frac{f(p)^n}{f(\alpha)^n}.
\label{eq:simplify_entropy}
\]
Now recall that by eq.~\eqref{eq:prob_of_winning_k_games}, $P(ab|xy)$ is related to $p_k$ by
\[
P(ab|xy) = \frac{p_k}{2^n\binom{n}{k}}.
\label{eq:chsh_box_from_probs}
\]
It is a well known identity of the Beta function that
\[
\binom{n}{k}^{-1} = (n+1)\int_0^1 t^k(1-t)^{n-k} \der t = 2^n (n+1)\int_0^1 f(p)^n \der p \leq 2^n (n+1) f(\alpha)^n
\label{eq:chsh_beta_function_identity}
\]
where for the last inequality we used lemma \ref{lem:chsh_concavity_argument} and the fact that the maximum of $f$ on $[0,1]$ is $f(\alpha)$. Inserting eq.~\eqref{eq:chsh_beta_function_identity} followed by eq.~\eqref{eq:thresh_summary}--\eqref{eq:simplify_entropy} into eq.~\eqref{eq:chsh_box_from_probs} gives
\[
P(ab|xy) \leq (n+1)p_kf(\alpha)^n \leq (n+1) \sup_{p \in [1-w,w]} f(p)^n.
\label{eq:chsh_P_bound}
\]
Combining eq.~\eqref{eq:chsh_tau_bound} and eq.~\eqref{eq:chsh_P_bound} yields $P(ab|xy) \leq (n+1)^2\tau(ab|xy)$, as desired.
\end{proof}

The arguments in the proof of theorem \ref{thm:de_finetti_chsh_symmetric} are not specific to CHSH symmetry. In fact, in appendix~\ref{sec:a_de_Finetti_theorem_for_general_symmetries} we show that we get such a de Finetti theorem whenever a threshold theorem analogous to theorem \ref{thm:threshold_unger} holds.

\subsection{The second de Finetti theorem}
\label{subsec:second_de_finetti_theorem}
The de Finetti theorem discussed in the previous section is similar to the de Finetti theorems for boxes shown in \cite{ArnonFriedman2015}; they show that the entries of some given box are \emph{upper bounded} by the entries of a de Finetti box. The original de Finetti theorems for random variables  and quantum states are of a different flavor: They show that the marginal on the first $k$ rounds of an arbitrary $n$-round permutation invariant state is \emph{close} to a de Finetti state if $k \ll n$. In this section we show a theorem of this type for CHSH symmetric boxes. We use the following distance measure on the space of boxes:
\begin{mydef}
Let $P_{A|X}$ and $Q_{A|X}$ be two boxes with the same input set $\X$ and output set $\A$. Their distance is
\[
\norm{P_{A|X}-Q_{A|X}} = \max_{x \in \X} \sum_{a \in \A} |P(a|x)-Q(a|x)|.
\]
\end{mydef}
This distance is just the $\ell^1$ distance of the probability distributions of $a$, maximized over the input $x$. \edit{}{To state the de Finetti theorem we need to introduce the notion of the marginal of an $n$-round box. In general, this marginal may not be well-defined without some kind of no-signaling condition across different rounds (since otherwise the output distribution of one round could potentially depend on the input in another round). However, it turns out that for CHSH symmetric boxes this is indeed well-defined, as we now show.
\begin{lem}
Let $P_{AB|XY}$ be an $n$-round CHSH symmetric quantum box and let $1 \leq k \leq n$ be an integer. 
Then the expression 
\[
P^k(a_1...a_k,b_1...b_k|x_1...x_k,y_1...y_k) \coloneqq \sum_{\substack{a_{k+1}...a_n\\b_{k+1}...b_n}}P(a_1...a_n,b_1...b_n|x_1...x_n,y_1...y_n)\label{eq:marginal}
\]
is independent of the choice of $x_{k+1}...x_{n}$ and $y_{k+1}...y_{n}$, and we shall refer to it as the \emph{marginal} of the first $k$ rounds. Furthermore, $P^k_{AB|XY}$ is a CHSH symmetric quantum box (of $k$ rounds).
\end{lem}
\begin{proof}
We shall use the notation $a=(a_1...a_k)$ and $a' = (a_{k+1}...a_{n})$, and define $b,b',x,x',y,y'$ analogously. To see that $P^k_{AB|XY}$ is independent of the choice of $x'$ and $y'$ we calculate
\begin{align}
P^k(ab|xy) &= \sum_{a'b'} P(aa',bb'|xx',yy') \nonumber\\ 
		   &= \sum_{a'b'} P(aa',b(b'\oplus x'y')|x0,y0) \nonumber\\ 
		   &= \sum_{a'b'} P(aa',bb'|x0,y0),
\end{align}
where in the second equality we used the CHSH symmetry of $P_{AB|XY}$ and in the third equality we shifted the summation variable $b'$ by $x'y'$. 

To see the CHSH symmetry of $P^k_{AB|XY}$, note that by CHSH symmetry of $P_{AB|XY}$,  $P(aa',bb'|x0,y0)$ only depends on $a\oplus b\oplus xy$ and $a' \oplus b'$. Hence $P^k(ab|xy)$ only depends on $a\oplus b\oplus xy$. Furthermore, the permutation invariance of $P_{AB|XY}$ immediately implies that $P^k_{AB|XY}$ is also permutation invariant, and hence we conclude that $P^k_{AB|XY}$ is CHSH symmetric. Finally, the fact that $P^k_{AB|XY}$ is a quantum box immediately follows from the fact that $P_{AB|XY}$ is quantum.
\end{proof}
}

Now we can state the de Finetti theorem:
\begin{thm}
\label{thm:de_finetti_theorem_2}
Let $P_{AB|XY}$ be an $n$-round CHSH symmetric quantum box \edit{. Denote by $P^{k}_{AB|XY}$ the marginal of the first $k$ rounds.}{and let $P^{k}_{AB|XY}$  be the marginal of the first $k$ rounds as defined in eq.~\eqref{eq:marginal}.} There is a $k$-round CHSH symmetric quantum de Finetti box $\tau_{AB|XY}$ such that
\[
\norm{P^k_{AB|XY}-\tau_{AB|XY}} \leq \left(C\sqrt{\ln(n/k)}+4\right) \sqrt{\frac{k}{n}} + \frac{4k}{n} = \mathcal{O}\left(\sqrt{\ln(n/k)\frac{k}{n}}\right)
\]
with $C = \frac{2}{\sqrt{2-\sqrt{2}}}\approx 2.6$.
\end{thm}
For the proof of theorem \ref{thm:de_finetti_theorem_2} we first note that the distance between two CHSH symmetric boxes is just the $\ell^1$ distance between the distributions of the wins and losses of the CHSH game, which are independent from the input into the box.
\begin{lem}
\label{lem:chsh_distance}
Let $P_{AB|XY}$ and $Q_{AB|XY}$ be CHSH symmetric $n$-round boxes and $W = A\oplus B \oplus XY \oplus \mathbf{1} \in \{0,1\}^n$ be the random variable that indicates in which rounds the CHSH game was won. Let $P_W$ and $Q_W$ the distribution of $W$. Then
\[
\norm{P_{AB|XY}-Q_{AB|XY}} = \norm{P_W-Q_W}_1.
\]
\end{lem}
\begin{proof}
For all $x,y$
\[
P(ab|xy) = 2^{-n}P_W(a\oplus b\oplus xy \oplus \mathbf{1})
\]
so
\[
\sum_{a,b}|P(ab|xy)-Q(ab|xy)| = \sum_{a,w}2^{-n}|P_W(w)-Q_W(w)| = \norm{P_W-Q_W}_1.
\]
\end{proof}

Another ingredient for the proof of theorem \ref{thm:de_finetti_theorem_2} is a bound on the $\ell^1$-distance between two binomial distributions:
\begin{lem}
\label{lem:binom_continuity}
Let $k \in \Integers$ and $p,q \in (0,1)$. Denote by $P=\mathrm{Binom}(k,p)$ and $Q = \mathrm{Binom}(k,q)$ the binomial distributions with $k$ trials and success probabilities $p$ and $q$. Then
\[
\norm{P-Q}_1 \leq 2\sqrt{\frac{n}{\min\{q,1-q\}}} |p-q|.
\]
\end{lem}
\begin{proof}
Denote by $P_0$ and $Q_0$ the Bernoulli distributions with success probability $p$ and $q$ respectively. We use Pinsker's inequality and the reverse Pinsker's inequality (Lemma 4.1 in \cite{Goetze2019}) to calculate
\begin{align}
\norm{P-Q}_1 &\leq \sqrt{2D(P||Q)} & \qquad \text{ Pinsker's inequality}\nonumber \\
& = \sqrt{2nD(P_0||Q_0)} & \qquad \text{ additivity of relative entropy} \nonumber \\
&\leq \sqrt{\frac{n\norm{P_0-Q_0}_1^2}{\min\{q,1-q\}}} & \qquad \text{ reverse Pinsker's inequality}\nonumber \\
&= 2\sqrt{\frac{n}{\min\{q,1-q\}}} |p-q|.
\end{align}
\end{proof}
Now we are ready to prove the de Finetti theorem:
\begin{proof}[Proof of theorem \ref{thm:de_finetti_theorem_2}]
Denote by $p_N$ the probability that Alice and Bob win exactly $N$ CHSH games on the box $P_{AB|XY}$. Then $W=A\oplus B \oplus XY \oplus \mathbf{1}$ has a permutation invariant distribution $P_W$ with
\[
\sum_{w: \sum_i w_i = N} P_W(w) = p_N.
\]
By the de Finetti theorem for random variables \cite{Diaconis1980} we have
\[
\norm{ P^k_W - \sum_{N=0}^n p_N \mathrm{Binom}\left(k,\frac{N}{n}\right)} \leq \frac{4k}{n}.
\label{eq:diaconis_freedman}
\]
where $P^k_W$ denotes the distribution of the first $k$ bits of $W$.

For $p \in [0,1]$ denote by \editB{}{$Q(p)_{\hat{A}\hat{B}|\hat{X}\hat{Y}}$} the single-round box with CHSH winning probability $p$. Then by lemma \ref{lem:chsh_distance} and eq.~\eqref{eq:diaconis_freedman}
\[
\norm{P^k_{AB|XY} - \sum_{N=0}^n p_N Q\left(\frac{N}{n}\right)^{\otimes k} } \leq \frac{4k}{n}
\label{bound1}
\]
The box $\sum_{N=0}^n p_N Q\left(\frac{N}{n}\right)^{\otimes k}$ is CHSH symmetric and de Finetti, but not quantum. The problems are the terms with $N > nw$ and $N < n(1-w)$, where $w = \frac{2+\sqrt{2}}{4}$.
We define a quantum CHSH symmetric de Finetti box as
\[
\tau_{AB|XY} = \sum_{N=0}^n p_N \begin{cases} Q(w)^{\otimes k} & \text{ if } N > wn \\ Q\left(\frac{N}{n}\right) ^{\otimes k}& \text{ if } N \in [(1-w)n, wn] \\ Q(1-w)^{\otimes k} & \text{ if } N < (1-w) n\end{cases}.
\]

We will show
\[
\norm{\tau_{AB|XY} - \sum_{N=0}^n p_N Q\left(\frac{N}{n}\right)^{\otimes k}} \leq \left(C\sqrt{\ln(n/k)}+4\right) \sqrt{\frac{k}{n}}.
\label{bound2}
\]
The statement of theorem \ref{thm:de_finetti_theorem_2} then follows by combining this bound and the bound in eq.~\eqref{bound1} using the triangle inequality.

Let $\delta>0$. We split the sum in the definition of $\tau_{AB|XY}$ to obtain
\begin{align}
\norm{\tau_{AB|XY} - \sum_{N=0}^n p_N Q\left(\frac{N}{n}\right)^{\otimes k}} & \leq \sum_{N \in [wn,n]} p_N \norm{Q(w)^{\otimes k} - Q\left(\frac{N}{n}\right)^{\otimes k}} \nonumber \\
&+   \sum_{N \in [0,(1-w)n]} p_N \norm{Q(1-w)^{\otimes k} - Q\left(\frac{N}{n}\right)^{\otimes k}}\\\nonumber
&\leq \sum_{N \in [wn, (w+\delta)n]} p_N \norm{Q(w)^{\otimes k} -Q(w+\delta)^{\otimes k}}\\\nonumber
&+ \sum_{N \in [(1-w-\delta)n, (1-w)n]} p_N \norm{Q(1-w)^{\otimes k} -Q(1-w-\delta)^{\otimes k}}\\\nonumber
&+2\sum_{N\in[0,(1-w-\delta)n]\cup[(w+\delta)n,n]} p_N
\end{align}
where in the last line we used $\norm{P_{AB|XY}-Q_{AB|XY}} \leq 2$ for all normalized boxes $P_{AB|XY}$ and $Q_{AB|XY}$. 
We start by bounding the terms with $N \in [wn, (w+\delta)n]$ and $N \in [(1-w-\delta)n, (1-w)n]$ using lemma \ref{lem:binom_continuity}:
\begin{align}
&\norm{Q(w)^{\otimes k} -Q(w+\delta)^{\otimes k}}\leq \frac{2}{\sqrt{1-w}}\sqrt{k}\delta = 2C\sqrt{k}\delta
\end{align}
where we used $\frac{2}{\sqrt{1-w}} = \frac{4}{\sqrt{2-\sqrt{2}}}=2C$.
Analogously we find
\[
\norm{Q(1-w)^{\otimes k} -Q(1-w-\delta)^{\otimes k}} \leq 2C\sqrt{k}\delta
\]
so that by $\sum_N p_N =1$
 \begin{align}
  &\sum_{N \in [wn, (w+\delta)n]} p_N \norm{Q(w)^{\otimes k} -Q(w+\delta)^{\otimes k}}\nonumber \\
  &+ \sum_{N \in [(1-w-\delta)n, (1-w)n]} p_N \norm{Q(1-w)^{\otimes k} -Q(1-w-\delta)^{\otimes k}}\nonumber \\
  &\leq 2C\sqrt{k}\delta.
 \end{align}

Now we turn to the terms with $N>(w+\delta)n$ and $N<(1-w-\delta)n$. By the threshold theorem for the CHSH game (theorem \ref{thm:threshold_unger}) it holds that
\begin{align}
\sum_{N > (w+\delta)n}p_N &\leq e^{-nD(w+\delta, 1-w-\delta||w,1-w)}\\ 
& \leq e^{-2n\delta^2}
\end{align}
where in the last step we used Pinsker's inequality, i.e $D(p,1-p||q,1-q) \geq 2 |p-q|^2$.
Analogously also
\[
\sum_{N<(1-w-\delta)n}p_N \leq e^{-2n\delta^2}.
\]
Putting together the bounds for $N \in [wn, (w+\delta)n]$ and $N>(w+\delta)n$ we find
\begin{align}
\norm{\tau_{AB|XY} - \sum_{N=0}^n p_N Q\left(\frac{N}{n}\right)^{\otimes k}}  &\leq 2C\sqrt{k}\delta+ 4e^{-2n\delta^2}.
\end{align}
Now we choose 
\[
\delta = \frac{1}{2}\sqrt{\frac{\ln(n/k)}{n}}.
\label{eq:choice_delta}
\]
and obtain
\[
\norm{\tau_{AB|XY} - \sum_{N=0}^n p_N Q\left(\frac{N}{n}\right)^{\otimes k}} \leq \left(C\sqrt{\ln(n/k)}+4\right) \sqrt{\frac{k}{n}}.
\]
This completes the proof.
\end{proof}
The choice of $\delta$ in eq.~\eqref{eq:choice_delta} is not optimal, it does not give the minimal possible error term in theorem \ref{thm:de_finetti_theorem_2}. However, the improvement that can be achieved by choosing $\delta$ optimally does not change the $\mathcal{O}\left(\sqrt{\ln(n/k) k/n}\right)$ behavior. To see this, choose 
\[
\delta = \frac{1/2\sqrt{\ln(n/k)}+\beta}{\sqrt{n}}
\]  
for some $\beta > -\sqrt{\ln(n/k)}/2$. Then the error term is given by
\[
\epsilon \coloneqq 2C\sqrt{k}\delta+ 4e^{-2n\delta^2} = \left( C\sqrt{\ln(n/k)} + 2C\beta +4e^{-2\beta^2-2\beta\sqrt{\ln(n/k)}}\right) \sqrt{k/n}.
\]
The optimal $\beta$ is such that $C' = 2C\beta +4e^{-2\beta^2-2\beta\sqrt{\ln(n/k)}}$ is minimized. A numerical optimization indicates that for $\ln(n/k)=0$ the minimum of $C'$ is achieved at $\beta = 0$, so $C'=4$. As $\ln(n/k)$ increases the minimum of $C'$ decreases slowly, at $\ln(n/k) = 10$ it is given by $C' \approx 2.03$, and at $\ln(n/k)=100$ by $C' \approx 0.96$. As $\ln(n/k) \rightarrow \infty$ it converges $C' \rightarrow 0$, which can be seen by choosing $\beta = \left(\ln(n/k)\right)^{-1/4}$. Regardless of the choice of $\beta$ the error is always at least $C\sqrt{\ln(n/k)k/n}$. 

\section{\edit{}{Applications}}
\label{sec:applications}

In this section we show how our first de Finetti theorem (theorem \ref{thm:de_finetti_chsh_symmetric}) has applications in DIQKD security proofs, \edit{}{by first using it to derive a bound on channel distinguishability, then discussing its implications for security proofs}. This result, and the proof of it, are analogous to theorem 25 in \cite{ArnonFriedman2015}, except that we use theorem \ref{thm:de_finetti_chsh_symmetric} as the de Finetti theorem, rather than the statement in eq.~\eqref{eq:de_finetti_theorem_general}. We remark that the works~\cite{Jain2020,Vidick2017} also used threshold theorems (of somewhat different forms) to obtain DIQKD security proofs. However, \editB{}{as discussed in the introduction}, their proof techniques currently only apply to protocols using one-way error correction, and yield lower asymptotic keyrates compared to the iid case. In contrast, the results we derive here could be applied to all protocols having the appropriate symmetry properties. While they currently do not yield a full reduction to the iid case, our hope is that it would be possible to develop them further to obtain security proofs that are more generally applicable and yield higher asymptotic keyrates compared to~\cite{Jain2020,Vidick2017}, as was the case for de Finetti theorems in device-dependent QKD~\cite{Christandl2009a}.

\subsection{\edit{}{Bound on the diamond distance between channels}}

Here we consider channels on boxes of the following form: A channel $\E$ that acts on boxes of the form $P_{A|X}$ and outputs a random variable $R$ as its result is described by a probability distribution $P^\E_{X}$ on $\X$, and a conditional probability distribution $P^\E_{R|AX}$ which determines the result $R$ given $A$ and $X$. When acting on $P_{A|X}$ the channel produces a distribution on $R$ given by
\[
\E(P_{A|X})(r) = \sum_{x,a} P^\E_{X}(x)P_{A|X}(a|x)P^\E_{R|AX}(r|ax).
\] 
This definition is general enough to capture all protocols in a parallel DIQKD scenario, where all bits of the $n$-bit input $\X$ are entered at the same time into the box. It does not cover all protocols that are possible in a sequential DIQKD scenario \cite{ArnonFriedman2018}, where some of the input bits are only given to the box after some output bits have been received. In such a sequential scenario it is in principle possible to construct channels where the input to the box in some round depends on the output of the box in previous rounds.

If we consider boxes $P_{AE|XZ}$, where the additional $E,Z$ interface is held by Eve, we can also apply the channel $\E$ only to the $A,X$ interface to obtain a box with input $Z$ and outputs $R$ and $E$. We will denote this box by $\left(\E\otimes\id\right)(P_{AE|XZ})_{RE|Z}$.

We define the distance between two channels $\E$ and $\F$ by how well Eve can distinguish the boxes $\left(\E\otimes\id\right)(P_{AE|XZ})_{RE|Z}$ and $\left(\F\otimes\id\right)(P_{AE|XZ})_{RE|Z}$ if she is also given access to $R$. Then she can choose her input $Z$ dependent on $R$. This leads to the following definition \cite{ArnonFriedman2015, Hanggi2010}:
\begin{mydef}
Let $\E$ and $\F$ be two channels acting on boxes of the form $P_{A|X}$. The distinguishablity of $\E$ and $\F$ using the box $P_{AE|XZ}$ is given by
\begin{align}
&\norm{\left(\E-\F\right)\otimes\id\left(P_{AE|XZ}\right)} \nonumber \\
&= \sum_{r}\max_z\sum_e \left|\sum_{a,x} P_{AE|XZ}(ae|xz) \left(P^\E_{X}(x)P^\E_{R|AX}(r|ax)-P^\F_{X}(x)P^\F_{R|AX}(r|ax) \right)\right|.
\end{align}
We define the diamond distance between the channels with respect to some set $\P$ of boxes to be the following:
\begin{align}
&||\E-\F||_\Diamond^\P = \sup_{P_{AE|XZ} \in \P} \norm{\left(\E-\F\right)\otimes\id\left(P_{AE|XZ}\right)}.
\end{align}
\end{mydef}
Similarly to the usual diamond distance between quantum channels, the above definition of diamond distance with respect to some set $\P$ is a measure of how distinguishable the channels are with respect to a distinguisher that can only use boxes from $\P$. Simple choices of $\P$ include for instance the sets of quantum or non-signaling boxes.
For the following main theorem of this section we will however take $\P$ to be the set of quantum boxes $P_{ABE|XYZ}$ such that the marginal $P_{AB|XY}$ has CHSH symmetry, and denote the diamond distance with respect to this $\P$ as $||\E-\F||_\Diamond^{\mathrm{quantum,CHSH}}$. Note that if the action of the channels $\E,\F$ can be described by Alice and Bob first performing the depolarizing procedure described above, this restriction causes no change in the diamond distance as compared to choosing $\P$ to be the entire set of quantum boxes $P_{ABE|XYZ}$.
\begin{thm}
Let $\E$ and $\F$ two channels on $n$-round boxes of the form $P_{AB|XY}$, and let $\tau_{AB|XY}$ be the de Finetti box from theorem \ref{thm:de_finetti_chsh_symmetric}. Then
\begin{align}
||\E-\F||_\Diamond^{\mathrm{quantum,CHSH}} \leq (n+1)^2 \sup_{\tau_{ABE|XYZ}} \norm{(\E-\F)\otimes\id(\tau_{ABE|XYZ})}
\label{eq:postselection_ns}
\end{align}
where the supremum is taken over all non-signaling boxes that have the marginal $\tau_{AB|XY}$.
\label{thm:postselection_ns}
\end{thm}

\begin{proof}[Proof of theorem \ref{thm:postselection_ns}]
Let $P_{ABE|XYZ}$ be a quantum box whose marginal $P_{AB|XY}$ has CHSH symmetry. Let $R_{AB|XY}$ be such that
\[
\tau_{AB|XY} = \frac{1}{(n+1)^2}P_{AB|XY} + \left(1-\frac{1}{(n+1)^2}\right)R_{AB|XY}.
\label{eq:definition_R}
\]
By theorem \ref{thm:de_finetti_chsh_symmetric} all entries of $R_{AB|XY}$ are positive. Because the non-signaling condition is linear and $\tau_{AB|XY}$ and $P_{AB|XY}$ are non-signaling, $R_{AB|XY}$ is also non-signaling. Now we define an extension $\tau_{ABE|XYZ}$ of $\tau_{AB|XY}$ as follows: The box has one more possible outcome for Eve then the box $P_{ABE|XYZ}$. We will call this additional outcome $e^*$. The box $\tau_{ABE|XYZ}$ then works as follows: With probability $(n+1)^{-2}$ the box acts just like $P_{ABE|XYZ}$, and with probability $1-(n+1)^{-2}$ it always returns $e^*$ to Eve and acts like $R_{AB|XY}$ for Alice and Bob. Formally, this is given by
\[
\tau(abe|xyz) = \begin{cases} \frac{1}{(n+1)^2}P(abe|xyz) & \text{ if } e \neq e^* \\ \left(1-\frac{1}{(n+1)^2}\right)R(ab|xy) &\text{ if } e =e^* \end{cases}.
\]
Since $\tau_{ABE|XYZ}$ is the linear combination of two non-signaling boxes it is non-signaling itself. Furthermore, by eq.~\eqref{eq:definition_R} it is an extension of $\tau_{AB|XY}$. Finally, it holds that
\begin{align}
&||(\E-\F)\otimes\id(\tau_{ABE|XYZ})|| \nonumber \\
&= \sum_r \max_z \sum_e \left| \sum_{a,b,x,y} \tau(abe|xyz) \left(P^\E(xy)P^\E(r|abxy) - P^\F(xy)P^\F(r|abxy)\right)\right|\nonumber \\
&\geq \sum_r \max_z \sum_{e \neq e^*} \left| \sum_{a,b,x,y} \tau(abe|xyz) \left(P^\E(xy)P^\E(r|abxy) - P^\F(xy)P^\F(r|abxy)\right)\right|\nonumber \\
&= (n+1)^{-2}\sum_r \max_z \sum_{e \neq e^*} \left| \sum_{a,b,x,y} P(abe|xyz) \left(P^\E(xy)P^\E(r|abxy) - P^\F(xy)P^\F(r|abxy)\right)\right|\nonumber \\
&=(n+1)^{-2}||(\E-\F)\otimes\id(P_{ABE|XYZ})||.
\end{align}
Hence for all $P_{ABE|XYZ}$
\[
||(\E-\F)\otimes\id(P_{ABE|XYZ})|| \leq (n+1)^2\sup_{\tau_{ABE|XYZ}} ||(\E-\F)\otimes\id(\tau_{ABE|XYZ})||.
\]
Taking the supremum over all $P_{ABE|XYZ}$ with CHSH symmetric marginal $P_{AB|XY}$ yields the claim.
\end{proof}

\subsection{\edit{}{Implications for DIQKD security proofs}}

Theorem \ref{thm:postselection_ns} can be seen as a version of the postselection theorem for quantum channels \cite{Christandl2009a}. It allows us to bound the distance between two channels by the distinguishability of the channels when Eve is restricted to extensions of a fixed de Finetti box. This could potentially be a useful tool in security proofs of DIQKD protocols, because a protocol can be defined to be secure if its diamond distance to an ideal protocol is small \cite{Maurer2011, Portmann2021}. \edit{}{In particular, Theorem~\ref{thm:postselection_ns} implies that to prove security against coherent quantum attacks, it is sufficient to prove security for the case where the marginal of Alice and Bob is given by $\tau_{AB|XY}$, and Eve possesses a non-signaling extension of this box. 
This helps to simplify the task of a DIQKD security proof, because it means that it suffices to analyze (extensions of) the \emph{specific} box $\tau_{AB|XY}$, which has the convenient property of being a convex combination of iid quantum boxes.}  

However, there is a caveat: Although the box $\tau_{AB|XY}$ is quantum, the extensions $\tau_{ABE|XYZ}$ in the theorem statement here are allowed to be general non-signaling boxes. Furthermore, we will show in the next section that an adversary who has access to arbitrary non-signaling extensions of $\tau_{AB|XY}$ can actually be strictly better at distinguishing channels than an adversary who has only access to collective attack boxes. Hence theorem \ref{thm:postselection_ns} does not immediately yield security against coherent attacks from security against collective attacks \edit{}{--- still, since it does allow a ``partial'' reduction to the latter (namely, allowing us to focus on extensions of a quantum de Finetti box $\tau_{AB|XY}$), it may still  simplify DIQKD security proofs.}

\edit{}{We also remark that for our second de Finetti theorem (theorem~\ref{thm:de_finetti_theorem_2}), we currently do not have in mind an explicit application of it in DIQKD security proofs. Still, we presented it in this work in case it has applications in other contexts --- for instance, it might be useful in proving properties that only depend on the box $P_{AB|XY}$ itself, rather than involving its extensions as in DIQKD security proofs. It is also more similar to the original de Finetti theorem for classical random variables, or the early versions for quantum states developed in e.g.~\cite{Caves2002}.}

\section{\editB{Impossibility of}{Difficulties in} bounding the diamond distance by restriction to collective attacks}
\label{sec:impossible_strengthening}

Theorem \ref{thm:postselection_ns} shows that to bound the diamond distance between two channels $\E$ and $\F$ it is sufficient to restrict the attacker to non-signaling extensions of a fixed de Finetti box. There are many desirable strengthenings of this result: For example, one could restrict the attacker only to quantum extensions of the de Finetti box. One could also further restrict the attacker to use only quantum extensions of iid boxes, instead of the fixed de Finetti box. Finally, one could also restrict the attacker to collective attack boxes (defined below), as would be desirable to conclude security against coherent attacks directly from security against collective attacks. In this section we will see that a theorem like theorem \ref{thm:postselection_ns} does not hold for this strongest restriction; \editB{}{more precisely, we show that it is impossible for the bound~\eqref{eq:postselection_ns} to hold if the supremum is instead restricted to collective attack boxes (which we define later below)}. It remains open whether such a theorem holds for one of the other strengthenings mentioned above, \editB{}{or whether a reduction to collective attacks in a somewhat different form is possible}. \edit{}{(We note that the answers to these questions do not straightforwardly follow from existing no-go theorems on non-signaling privacy amplification \cite{Haenggi2013, ArnonFriedman2012}, since in our result $\tau_{AB|XY}$ is restricted to a convex combination of quantum distributions rather than non-signaling distributions.)}

We start by defining the boxes that an attacker is allowed to use in collective attacks. 
While there is potentially some flexibility in defining this, here we use a definition that essentially corresponds to the boxes considered in the security proofs of~\cite{Acin2007,Tan2020}, up to a collective measurement on Eve's side-information:
\begin{mydef}
\label{def:collective_attack_box}
An $n$-round quantum box $P_{ABE|XYZ}$ is a collective attack box if there are
\begin{itemize}
\item single round Hilbert spaces $\H_A$ and $\H_B$ for Alice and Bob and a Hilbert space $\H_E$ for Eve and 
\item a state $\rho_{ABE} \in \mathrm{End}\left(\H_A^{\otimes n} \otimes \H_B^{\otimes n} \otimes \H_E\right)$ such that the marginal $\rho_{AB}$ is iid and
\item for each $x$ a POVM $\{E^{a,x} \in \mathrm{End}(\H_A)|a\}$ \edit{}{on $\H_A$}, for each $y$ a POVM $\{F^{b,y} \in \mathrm{End}(\H_B)|b\}$  \edit{}{on $\H_B$} and for each $z$ a POVM $\{G^{e,z} \in \mathrm{End}(\H_E)|e\}$  \edit{}{on $\H_E$}
\end{itemize}
such that
\[
P(abe|xyz) = \tr\left[\rho_{ABE} \left(E^{a_1,x_1} \otimes ... \otimes E^{a_n,x_n} \otimes F^{b_1,y_1} \otimes ... \otimes F^{b_n,y_n} \otimes G^{e,z} \right) \right].
\] 
\end{mydef}
We remark on two aspects of the above definition. Firstly, note that we assume the Hilbert spaces of Alice and Bob can be split into $n$ rounds, but assume no internal structure of Eve's Hilbert space. However, since the state $\rho_{AB}$ is iid and thus has an iid purification, the state $\rho_{ABE}$ is related by a local operation on Eve's system to this iid purification. Since we assume nothing about $G^{e,z}$ except that it is a valid POVM, we can absorb this local operation into $G^{e,z}$ and thus describe any collective attack box also with a state $\rho_{ABE}$ that is iid. Collective attack boxes can thus be seen as boxes that are essentially iid, up to Eve performing a local operation on her systems followed by a joint measurement. Secondly, the fact that the definition inherently incorporates this measurement means that Eve's system is forced to be a box rather than a genuine quantum state. However, for the purposes of computing diamond distance, this in fact does not make a difference (as long as arbitrary POVMs $G^{e,z}$ are allowed in the definition) --- observe that the process of a distinguisher producing a guess for the channel can be described as it performing a POVM on its systems, and the optimal such POVM essentially induces a valid choice of $G^{e,z}$ in the above definition.

A crucial observation on collective attack boxes is the following: Consider the box $P^{e,z}_{AB|XY}$ which described the outcomes of Alice and Bob conditioned on Eve inputting $z$ and getting outcome $e$. It is given by
\begin{align}
P^{e,z}_{AB|XY}(ab|xy) &= \frac{P_{ABE|XYZ}(abe|xyz)}{P_{E|Z}(e|z)} \nonumber \\
&=  \tr\left[\rho^{e,z}_{AB} \left(E^{a_1,x_1} \otimes ... \otimes E^{a_n,x_n} \otimes F^{b_1,y_1} \otimes ... \otimes F^{b_n,y_n} \right) \right]
\end{align}
where $\rho^{e,z}_{AB}$ is a valid state,
\[
\rho^{e,z}_{AB} = \frac{\tr_E\left[\rho_{ABE}\left(\id_A\otimes\id_B\otimes G^{e,z}\right)\right]}{\tr\left[\rho_{ABE}\left(\id_A\otimes\id_B\otimes G^{e,z}\right)\right]}.
\]

Because $\sum_a E^{a,x} = \sum_b F^{b,y} = \id$ we see that $P^{e,z}_{AB|XY}$ is not only non-signaling between Alice and Bob, but also between the individual rounds. This means that for example
$\sum_{a_1} P^{e,z}_{AB|XY}(a_1a_2...a_nb|xy)$ does not depend on $x_1$.

The following main result of this section exploits this insight:
\begin{thm}
\label{thm:no_postselection_for_collective_attacks}
For each $n>1$ there exist two channels $\E$ and $\F$ acting on $n$-round boxes of the form $P_{AB|XY}$ such that $||(\E-\F)\otimes\id(P_{ABE|XYZ})|| = 0$ for all collective attack boxes $P_{ABE|XYZ}$, but $||\E-\F||_\Diamond^{\mathrm{quantum,CHSH}} \neq 0$. 
\end{thm} 
Theorem \ref{thm:no_postselection_for_collective_attacks} shows that a statement like theorem \ref{thm:postselection_ns} cannot hold if we maximize only over collective attack boxes instead of all non-signaling extensions of the fixed de Finetti box (not even for, say, an exponential prefactor instead of $(n+1)^2$). This shows that an attacker who has access to any non-signaling extension of the fixed de Finetti box is stronger than an attacker who has only access to collective attack boxes.

In the proof of theorem \ref{thm:no_postselection_for_collective_attacks} we will use that all collective attack boxes are non-signaling between the rounds of Alice and Bob, and that the non-signaling condition is linear. The following lemma will be crucial. It states that for each linear subspace of the probability distributions on some set, there are two channels (which act on probability distributions, not yet on boxes), that cannot be distinguished by any probability distribution in the linear subspace:
\begin{lem}
Let $\X$ be some finite set. We treat the unnormalized probability distributions on $\X$ as an orthant of an $|\X|$ dimensional real vector space. Let $\P$ be some linear subspace in this vector space, and $Q = (Q(x))_{x\in\X} \not\in \P$. Then there are two conditional probability distributions $P^\E_{R|X}$ and $P^\F_{R|X}$ such that the following holds: Denote for a probability distribution $P$ on $\X$ by $\E(P)$ and $\F(P)$ the distributions on $\R$ which are obtained by first sampling $x$ using $P$ and the sampling $r$ using  $P^\E_{R|X}$ and $P^\F_{R|X}$, i.e.~$\E(P)(r)=\sum_x P(x)P^\E(r|x)$. Then
\[
\E(P)=\F(P)
\]
for all $P \in \P$ and
\[
\E(Q) \neq \F(Q).
\]
\label{lem:instistinguishable_classical_channels}
\end{lem}
\begin{proof}
There exists a vector $\Delta = (\Delta_x)_{x\in\X}$ with $|\Delta_x| \leq 1$ for all $x$ and $\Delta \cdot P = 0$ for all $P \in \P$ and $\Delta \cdot Q \neq 0$. Take $\R = \{0,1\}$ and
\begin{align}
P^\E(0|x) &= \frac{1+\Delta_x}{2}\\
P^\E(1|x) &= \frac{1-\Delta_x}{2}\\
P^\F(0|x) &= \frac{1-\Delta_x}{2}\\
P^\F(1|x) &= \frac{1+\Delta_x}{2}.
\end{align}
Then for $P$ any probability distribution on $\X$
\begin{align}
||\E(P)-\F(P)||_1 &= \sum_r \left| \sum_x P(x) (P^\E(r|x) - P^\F(r|x)) \right| \nonumber \\
 &= |P\cdot \Delta| + |P \cdot (-\Delta)|\nonumber \\
&= 2|P\cdot \Delta|.
\end{align}
Hence for all $P \in \P$
\[
||\E(P)-\F(P)||_1=0
\]
and
\[
||\E(Q)-\F(Q)||_1\neq 0.
\]
\end{proof}

We will now construct two channels $\E$ and $\F$ that act on boxes $P_{A|X}$ (i.e.~we consider only Alice) that cannot be distinguished by any collective attack box, but that can be distinguished by a certain quantum box that is not a collective attack box. We will then see how to modify this construction to include Bob and to ensure that the box used to distinguish both channels has CHSH symmetry on the marginal of Alice and Bob.

\begin{lem}
\label{lem:no_postselection_theorem_single_round_boxes}
For each $n>1$ there are two channels $\E$ and $\F$ that act on $n$-round boxes $P_{AE|XZ}$ such that
$||(\E-\F)\otimes\id(P_{AE|XZ})|| = 0$ for all collective attack boxes $P_{AE|XZ}$, but $||\E-\F||_\Diamond^{\mathrm{quantum,CHSH}} \neq 0$
\end{lem}

\begin{proof}
We first construct the channels $\E$ and $\F$, then show that Eve cannot distinguish them if she is restricted to collective attack boxes, and finally show that there is a quantum box (naturally not a collective attack box) that can be used to distinguish both channels. We construct the channels $\E$ and $\F$ as follows, depending on a parameter $m>n/2$. For both channels, Alice does the following steps:
\begin{enumerate}
\item She enters uniformly random inputs $x_1,...,x_n$ into the inputs of her box.
\item She collects the outputs $a_1,...,a_n$.
\item She calculates
\[
t = \sum_{i=1}^m x_i
\]
and 
\[
w = \sum_{i=m+1}^n a_i.
\]
\label{last_step}
\end{enumerate}
Now consider the linear subspace $\P$ of probability distributions on the $(w,t)$ given by the linear constraints
\[
\frac{P(w,t)}{2^{-m}\binom{m}{t}} = \frac{P(w,t')}{2^{-m}\binom{m}{t'}}
\label{eq:no_ps_lin_constraint}
\]
for all $w,t,t'$ and take a $Q \not\in \P$ (a specific $Q$ will be constructed below). Alice constructs the channels $\E$ and $\F$ by applying the conditional probability distributions $P^\E_{R|WT}$ and $P^\F_{R|WT}$ from lemma \ref{lem:instistinguishable_classical_channels} to her result $(w,t)$ from step \ref{last_step}.

Now we prove that Eve cannot distinguish $\E$ and $\F$ if she uses a collective attack box $P_{AE|XZ}$. For this, we use that for all $e$ and $z$ $P^{e,z}_{A|X}$ is non-signaling between the rounds of Alice. In particular, the outputs of the rounds $m+1,...n$ cannot depend on the inputs in the rounds $1,...,m$, so $W$ and $T$ are independent when generated using $P^{e,z}_{A|X}$. Hence
\[
\frac{P^{e,z}_{WT}(wt)}{{2^{-m}\binom{m}{t}} } = P^{e,z}_{W|T}(w|t) = P^{e,z}_{W|T}(w|t') = \frac{P^{e,z}_{WT}(wt')}{{2^{-m}\binom{m}{t'}} } 
\]

The box $P^{e,z}$ therefore satisfies eq.~\eqref{eq:no_ps_lin_constraint}, and hence $||(\E-\F)(P^{e,z}_{A|X})||=0$ by lemma \ref{lem:instistinguishable_classical_channels}. Since this holds for all $e$ and $z$ we have also $||(\E-\F)\otimes \id(P_{AE|XZ})|| = 0$.

Finally we construct a box $Q_{A|X}$ that allows Eve to distinguish the channels $\E$ and $\F$ with a nonzero advantage over guessing. Notice that here Eve does not keep any system (neither quantum not classical) for herself and can distinguish the channels only from their result $R$. $Q_{A|X}$ can then be an arbitrary conditional probability distribution. Take $Q_{A|X}$ such that the result is surely $a=(1,1,...,1)$ if $\sum_{i} x_i > n/2$ and surely $a=(0,0,...,0)$ otherwise.
Then if $t > n/2$
\[
\frac{Q_{WT}(n-m,t)}{2^{-m}\binom{m}{t}} = Q_{W|T}(n-m|t) = 1
\] 
and if $t < m-n/2$
\[
\frac{Q_{WT}(n-m,t)}{2^{-m}\binom{m}{t}} = Q_{W|T}(n-m|t) = 0
\] 
so eq.~\eqref{eq:no_ps_lin_constraint} does not hold.
\end{proof}
Now we can adapt the statement of lemma \ref{lem:no_postselection_theorem_single_round_boxes} to prove theorem \ref{thm:no_postselection_for_collective_attacks}.
\begin{proof}[Proof of theorem \ref{thm:no_postselection_for_collective_attacks}]
First we generalize the construction in lemma \ref{lem:no_postselection_theorem_single_round_boxes} to boxes for which also Bob has an input, i.e.~boxes of the form $P_{AB|XY}$. For this, we take the channels $\E$ and $\F$ such that they act like in lemma \ref{lem:no_postselection_theorem_single_round_boxes} on Alice's inputs and outputs, and give an arbitrary input $Y$ and ignore the output $B$ for Bob. Clearly, both channels still cannot be distinguished with collective attack boxes, but can be distinguished by a box $Q_{AB|XY}$, which acts like the box $Q_{A|X}$ from lemma \ref{lem:no_postselection_theorem_single_round_boxes} on $A$ and $X$ and arbitrarily on $B$ and $Y$. However, $Q_{AB|XY}$ does not have CHSH symmetry. Using the depolarizing procedure in \cite{Masanes2006} we can construct a box $\tilde{Q}_{ABE|XY}$ such that the marginal $\tilde{Q}_{AB|XY}$ has CHSH symmetry and there is an output $e^*$ for Eve (corresponding to the case in which the depolarizing protocol does nothing), such that $\tilde{Q}^{e^*}_{AB|XY} = Q_{AB|XY}$. Then to distinguish $\E$ and $\F$ using $\tilde{Q}_{ABE|XY}$ Eve first checks $E$. If $E=e^*$ she distinguishes $\E$ and $\F$ as in lemma \ref{lem:no_postselection_theorem_single_round_boxes}, otherwise she just guesses randomly. Because the probability that $E=e^*$ is nonzero, we have $\norm{(\E-\F)\otimes\id(\tilde{Q}_{ABE|XY})} > 0$.
\end{proof}

Several remarks are in order. Firstly, the diamond norm $||\E-\F||_\Diamond^{\mathrm{quantum,CHSH}}$ between the channels $\E$ and $\F$ from theorem \ref{thm:no_postselection_for_collective_attacks} is exponentially small in $n$, because Eve only tries to distinguish $\E$ and $\F$ in the case when the depolarizing protocol does nothing. This means that while a theorem in the same form as theorem \ref{thm:postselection_ns} cannot hold for collective attack boxes, it is entirely possible that there is, for example, a theorem that yields a bound of the form
\[||\E-\F||_\Diamond^{\mathrm{quantum,CHSH}} \leq  f(n) \sup_{P_{ABE|XYZ}} \norm{(\E-\F)\otimes\id(P_{ABE|XYZ})} + g(n),
\] 
where $g(n)\rightarrow 0$ as $n\rightarrow \infty$.
Such a result could be sufficient to allow security proof reductions to collective attacks.

Secondly, the results of this section relied only on the fact that collective attack boxes are non-signaling between the rounds. This property arose entirely from the fact that Alice and Bob's measurements act on different Hilbert spaces in different rounds, and hence also holds more generally, i.e.~even if the measurements in each round are different, or the states are entangled across rounds. This seems to suggest that in DIQKD, imposing an assumption that different rounds have different Hilbert spaces may already be a fairly strong restriction by itself\footnote{For comparison, in device-dependent QKD, this assumption is often implicitly imposed by default. This perhaps suggests a possible source of what appear to be greater challenges in non-iid DIQKD security proofs as compared to device-dependent QKD. It may also indicate that the default assumptions in device-dependent QKD could be stronger than they initially appear.}, even if we allow many other non-iid behaviours across states in different rounds, such as classical correlations or even entanglement.\edit{}{(In fact, security proof reductions to the iid case under this assumption were indeed previously studied in~\cite{Haenggi2010,Masanes2011}, though the latter was restricted to one-way protocols.)} Whether this assumption seems reasonable may depend on the protocol --- for instance, it seems unsatisfactory if each honest party has to use a single device for all inputs/outputs (as in~\cite{ArnonFriedman2018}, which used the EAT to avoid this assumption for one-way protocols), but if each honest party has access to $n$ devices that are ``well isolated'' from each other, it might be more plausible.

Thirdly, we remark that all sequential DIQKD protocols naturally fulfill a certain form of non-signaling constraints between the individual rounds of Alice and Bob: Alice and Bob's inputs in one round cannot influence the outputs in preceding rounds. Theorem \ref{thm:no_postselection_for_collective_attacks} does not rule out that a result like theorem \ref{thm:postselection_ns} exists for channels with such a sequential structure. However, \edit{}{preserving} such a sequential structure \edit{}{for the purposes of a security proof} appears rather incompatible with permutation symmetry, so exploiting such sequential structures might require different techniques from those used in this paper.

\section{Conclusion}
\label{sec:conclusion}
In this paper we proved two de Finetti theorems for quantum conditional probability distributions with CHSH symmetry. The advantage of these theorems over similar de Finetti theorems \cite{ArnonFriedman2015,Christandl2009} is that the de Finetti boxes are in the quantum set. The first de Finetti theorem states that the entries of a CHSH symmetric box are upper bounded, up to a polynomial factor in $n$, by the entries of a fixed de Finetti box. This theorem is actually not restricted to boxes with CHSH symmetry but can be applied to arbitrary symmetries if a corresponding threshold theorem is available. The second de Finetti theorem states that the marginal of the first $k$ rounds of an $n$ round CHSH symmetric box is close to (and not just upper bounded by) a de Finetti box.

We further showed that the first de Finetti theorem can be used to obtain a bound on the diamond distance between two channels acting on boxes. Specifically, an attacker who tries to distinguish two channels $\E$ and $\F$ can be restricted to non-signaling extensions of a fixed quantum de Finetti box without decreasing the distinguishability between both channels by more than a polynomial factor. Because the security of DIQKD protocols is defined in terms of the distance between the channel given by the protocol and an ideal channel this statement might be useful in security proofs. However our theorem does not immediately allow to conclude security against coherent attacks from security against collective attacks: A straightforward strengthening of it to bound the diamond distance between two channels by the distinguishability using collective attack boxes does not hold. Based on some insights in our proof approach, we speculate that in DIQKD, assuming that boxes in different rounds act on different Hilbert spaces may already be a fairly strong constraint, even if we allow correlations or entanglement between states in different rounds.

\section*{Acknowledgements}

We thank Renato Renner for useful suggestions and feedback on this project, as well as Srijita Kundu for guidance on threshold theorems. \edit{}{We also thank the reviewers for helpful feedback in improving the manuscript, including drawing our attention to the results in~\cite{Haenggi2010,Masanes2011} for boxes satisfying the constraint that measurements in different rounds commute.}

\section*{Funding}

This project was funded by the Swiss National Science Foundation via the National Center for Competence in Research for Quantum Science and Technology (QSIT), the Air Force Office of Scientific Research (AFOSR) via grant FA9550-19-1-0202, and the QuantERA project eDICT.

\appendix

\section{Proof of lemma \ref{lem:chsh_concavity_argument}}
\label{sec:proof_concavity_lemma}
\begin{proof}[Proof of lemma \ref{lem:chsh_concavity_argument}]
The second inequality follows directly because $f$ is non-negative and attains its maximum at $x^*$. The idea for the first inequality is to replace $f$ by a piecewise linear function that equals 0 at $a$ and $b$, and equals $f(x^*)$ at $x^*$. By concavity this piecewise linear function is always smaller than $f$ itself.  Writing this out explicitly: By concavity and non-negativity we have for all $x \in [a,x^*]$ 
\begin{align}
f(x) &= f \left(\frac{x-a}{x^*-a}x^* + \frac{x^*-x}{x^*-a}a\right)\nonumber \\
&\geq \frac{x-a}{x^*-a}f(x^*)+\frac{x^*-x}{x^*-a}f(a)\nonumber \\
&\geq \frac{x-a}{x^*-a}f(x^*).
\end{align}

Therefore we have
\[
\int_a^{x^*} f(x)^n \der x \geq f(x^*)^n \int_a^{x^*} \left(\frac{x-a}{x^*-a}\right)^n\der x = \frac{1}{n+1}(x^*-a)f(x^*)^n.
\]
Analogously it follows
\[
\int_{x^*}^b f(x)^n \der x \geq \frac{1}{n+1}(b-x^*)f(x^*)^n,
\]
so together
\[
\int_a^b f(x)^n\der x = \int_a^{x^*} f(x)^n \der x + \int_{x^*}^b f(x)^n \der x \geq \frac{1}{n+1}(b-a)f(x^*)^n.
\]
\end{proof}

\section{A de Finetti theorem for general symmetries}
\label{sec:a_de_Finetti_theorem_for_general_symmetries}
In this section we will show that for arbitrary games a threshold theorem such as theorem \ref{thm:threshold_unger} can always be used to prove a de Finetti theorem similar to theorem \ref{thm:de_finetti_chsh_symmetric}. Conversely, we will also see that a de Finetti theorem implies a threshold theorem. This means that proving a de Finetti theorem for some symmetry is just as hard as proving a threshold theorem for the game associated with that symmetry.

\subsection{Statement of the main theorem}
Throughout this section we will only consider boxes with a single interface and with a single round input set \editB{}{$\Xh$}, a single round output set \editB{}{$\Ah$}, \editB{}{and corresponding $n$-round input and out sets $\X=\Xh^n$ and $\Y=\Yh^n$}. CHSH symmetric boxes can be described like this by treating the two parties Alice and Bob as one, so the input and output sets are \editB{}{$\Ah=\Xh=\{0,1\}^2$}. We consider the following generalization of CHSH symmetry:
\begin{mydef}
\begin{enumerate}
\item Let $d \in \Integers$ and let \editB{}{$w:\Ah\times\Xh\rightarrow \{1,\dotsc, d\}$} be some function. We will call $w$ the predicate function of the symmetry.  For \editB{}{$a\in\A$} and \editB{}{$x\in\X$} we define $\freq^w(a,x) = (\frac{k_1}{n},\dots,\frac{k_d}{n}) \in \simplex{d}$ with $k_r = | \left\lbrace i | w(a_i,x_i)=r \right\rbrace |$ . Here $\simplex{d}$ denotes the $d$-dimensional simplex. 
\item We say an $n$-round box $P_{A|X}$ has $w$-symmetry if $P(a|x)=P(a'|x')$ whenever $\freq^w(a,x)=\freq^w(a',x')$, for all $a,a' \in \A^n$ and $x,x' \in \X^n$. 
\end{enumerate}
\end{mydef}
CHSH symmetry is an example of $w$-symmetry with $w((a,b), (x,y))=1$ if $a\oplus b = xy$ and $w((a,b), (x,y))=2$ if $a\oplus b \neq xy$. The definition of $w$-symmetry is an extension of the symmetries considered in \cite{ArnonFriedman2015} for permutation invariant boxes, where only certain predicate functions $w$ where considered, namely those where for each pair $x$, $x'$ either the images of $w(\cdot, x)$ and $w(\cdot, x')$ are disjoint or $w(\cdot,x)$ and $w(\cdot,x')$ are identical up to a permutation of the elements of \editB{}{$\Ah$}.

Instead of the set of quantum single round CHSH boxes, we will in this section consider a general convex set $\Q$ of single round boxes $Q_{\hat{A}|\hat{X}}$. If we view $\Q$ as a convex subset in \editB{}{$\Reals^{|\Ah||\Xh|}$} we can consider its affine hull: the smallest affine superset of $\Q$. Throughout this section we will denote the dimension of the affine hall by $d'$. For the CHSH symmetric case, $\Q$ is the set of CHSH symmetric quantum boxes, and $d'=1$.

In the CHSH symmetric case it was crucial that the expected number of wins of $n$-rounds of the CHSH games is between $\frac{2-\sqrt{2}}{4}$ and $\frac{2+\sqrt{2}}{4}$ when the games are played with iid quantum boxes. In our generalization the interval $[\frac{2-\sqrt{2}}{4},\frac{2+\sqrt{2}}{4}]$ will be replaced by a set of expected frequencies $\F_\mu$:
\begin{mydef}
Let $\mu$ be a probability distribution on $\X$, $\Q$ a convex set of single round boxes and $w$ a predicate function. The set of expected frequencies is
\[
\mathcal{F}_\mu = \left\lbrace \left( \left. \sum_{\substack{a,x \\ w(a,x)=r}}Q(a|x)\mu(x) \right)_{r=1...d} \right| Q_{\editB{}{\hat{A}|\hat{X}}} \in \mathcal{Q} \right\rbrace \subseteq \simplex{d}.
\]
\end{mydef}
In the CHSH symmetric case it is $\F_\mu = \{(p,1-p)|p\in [(2-\sqrt{2})/4, (2+\sqrt{2})/4]\}$ regardless of $\mu$.
We can now state the main theorem of this section:
\begin{thm}
Let \editB{}{$w:\Ah\times\Xh \rightarrow [d]$} be a predicate function and $\Q$ a convex subset of single round boxes with an affine hull of dimension $d'$. Let $\F_\mu$ be the set of expected frequencies.
There exists a de Finetti state $\tau_{A|X} \in \mathrm{conv}\left(\left\lbrace Q_{A|X}^{\otimes n} | Q_{\editB{}{\hat{A}|\hat{X}}} \in \Q\right\rbrace\right)$ independent of $w$ such that the following hold:
\begin{enumerate}
\item Let $P_{A|X}$ be an $n$-round box, let $\mu$ be a probability measure on $\X$, and take any $f\in \simplex{d}$. Suppose there exists some $C>0$ such that
\[
\Pr_{P_{A|X},\mu^{\otimes n}}\left[ \freq^w(A,X) = f \right] \leq C\exp\left(-\inf_{f' \in \mathcal{F}_{\mu}} D(f||f') n\right).
\label{eq:threshold_thm}
\]
Then 
\[
\Pr_{P_{A|X},\mu^{\otimes n}}\left[ \freq^w(A,X)=f\right] \leq C \binom{n+d'}{d'} (n+1)^{d-1} \Pr_{\tau_{A|X},\mu^{\otimes n}}\left[ \freq^w(A,X)=f\right].
\]

\item Let $P_{A|X}$ be an $n$-round box with $w$-symmetry and let each box $Q_{\editB{}{\hat{A}|\hat{X}}} \in \Q$ have $w$-symmetry. Let $C$ be such that for all $f \in \simplex{d}$ there is a $\mu>0$ such that eq.~\eqref{eq:threshold_thm} holds. Then
\[
P(a|x) \leq C \binom{n+d'}{d'} (n+1)^{d-1} \tau(a|x) \qquad \forall a\in\A^n, x\in\X^n.
\label{eq:de_finetti_theorem}
\]

\item Let $P_{A|X}$ be an $n$-round box for which eq.~\eqref{eq:de_finetti_theorem} holds. Then
\[
\Pr_{P_{A|X},\mu^{\otimes n}}\left[ \freq^w(A,X) = f \right] \leq C \binom{n+d'}{d'} (n+1)^{d-1} \exp\left(-\inf_{f' \in \mathcal{F}_{\mu}} D(f||f') n\right).
\]
\end{enumerate}
\label{thm:deFinetti_main}
\end{thm}
Qualitatively, we can interpret the equations and statements in theorem \ref{thm:deFinetti_main} as follows: 
\begin{itemize}
\item The condition~\eqref{eq:threshold_thm} is a perfect threshold theorem, written in a form similar to eq.~\eqref{eq:thresh_summary} (which was for the CHSH case). It states that the probability to obtain a frequency distribution $f$ outside of the set $\F_\mu$ of expected frequencies decays exponentially with $n$ and with the distance from $f$ to $\F_\mu$, as measured by the relative entropy. 

\item Part 1 of theorem \ref{thm:deFinetti_main} states that if we have such a threshold theorem, 
then the probability of obtaining the frequencies $f$ using the box $P_{A|X}$ can be bounded by the probability of obtaining $f$ using the de Finetti box $\tau_{A|X}$, up to a polynomial factor. We have chosen to state this part of the theorem separately because it does not require $P_{A|X}$ to be $w$-symmetric. 

\item Part 2 asserts that if we have the further condition that $P_{A|X}$ and all boxes in $\Q$ are $w$-symmetric, then we can get a de Finetti theorem analogous to theorem \ref{thm:de_finetti_chsh_symmetric} in the CHSH case. We will derive part 2 from part 1 by expressing the entries of $P_{A|X}$ and $\tau_{A|X}$ in terms of their respective probabilities of obtaining $\freq^w(a,x)=f$, which is possible by $w$-symmetry. 

\item Finally part 3 shows the other direction of the equivalence between a threshold theorem and a de Finetti theorem: A box $P_{A|X}$ satisfying the de-Finetti theorem statement~\eqref{eq:de_finetti_theorem} also satisfies a threshold theorem, albeit with a larger prefactor than in eq.~\eqref{eq:threshold_thm}.
\end{itemize}

To prove theorem \ref{thm:deFinetti_main}, we first show some preparatory lemmas in sections \ref{subsec:define_de_Finetti_box} and \ref{subsec:bound_pax_by_iid}, then combine them in section \ref{subsec:proof_main_definetti}. 
Some insight can be gained into the proof structure by writing a slightly different proof of theorem \ref{thm:de_finetti_chsh_symmetric} (the CHSH case), in order to draw analogies to parts 1 and 2 of theorem \ref{thm:deFinetti_main} separately. This version of the proof proceeds as follows: first prove eq.~\eqref{eq:chsh_tau_bound} as before, giving a lower bound on $\tau(a|x)$.
However, we reorder the proof after that point. Namely, observe that combining~\eqref{eq:thresh_summary}, \eqref{eq:simplify_entropy} and \eqref{eq:chsh_beta_function_identity} gives
\[
p_k \leq (n+1) 2^n \binom{n}{k} \sup_{p\in[1-w,w]} f(p)^n.
\label{eq:chsh_bound_p_by_iid}
\]
Putting together~\eqref{eq:chsh_tau_bound} and~\eqref{eq:chsh_bound_p_by_iid} gives
\[
p_k \leq (n+1)^2 2^n \binom{n}{k} \tau(ab|xy).
\label{eq:chsh_freq_bounds}
\]
The symmetry condition has not been used up to this point. We now use it to relate $p_k$ to $P(ab|xy)$ via~\eqref{eq:prob_of_winning_k_games}, which yields the desired inequality $P(ab|xy) \leq (n+1)^2 \tau(ab|xy)$.

The proof in the subsequent sections basically follows the same structure as the above version.
First, eq.~\eqref{eq:chsh_tau_bound} is generalized to lemma~\ref{lem:def_tau} in section~\ref{subsec:define_de_Finetti_box}. Next, eq.~\eqref{eq:chsh_bound_p_by_iid} is replaced by lemma~\ref{lem:pBoundiid} in section~\ref{subsec:bound_pax_by_iid}, bounding the probability of obtaining some outcome frequency in terms of a supremum over iid boxes (the $2^n \binom{n}{k}$ factor in~\eqref{eq:chsh_bound_p_by_iid} counts different ways to achieve the specified frequency). These lemmas are combined to obtain part 1 of theorem \ref{thm:deFinetti_main}, which is the generalization of~\eqref{eq:chsh_freq_bounds}. Finally, the symmetries are invoked to relate the box distribution to the probabilities of obtaining some outcome frequencies, analogous to~\eqref{eq:prob_of_winning_k_games}, to yield part 2 of theorem \ref{thm:deFinetti_main}.

\subsection{Construction and properties of the de Finetti box}
\label{subsec:define_de_Finetti_box}
In this section we will construct the de Finetti box $\tau_{A|X}$ and show that this $\tau(a|x)$ is at most polynomially smaller then $Q^{\otimes n}(a|x)$, for all $Q\in\Q$. Before that we need to prove some preparatory lemmas:

The following lemma and proof are adopted from \cite{Sarwate2012}.
\begin{lem}[Matrix Determinant Lemma]
Let $A\in \Reals^{n\times n}$ be an invertible matrix and $v \in \Reals^n$. Then
\[
\det(A-vv^T) = \det(A)(1-v^TA^{-1}v)
\] 
\label{lem:matrix_determinant_lemma}
\end{lem}
\begin{proof}
We calculate $\det(A-vv^T)$ as the determinant of a block matrix:
\begin{align}
\det(A-vv^T) &= \det\left(\begin{matrix} A & v \\ v^T & 1 \end{matrix}\right) \nonumber \\
&= \det\left[\left(\begin{matrix}A & 0 \\ v^T & 1\end{matrix}\right)\left(\begin{matrix}I & A^{-1}v \\ 0& I-v^TA^{-1}v\end{matrix}\right)\right]\nonumber \\
&=\det(A)(1-v^TA^{-1}v)
\end{align}
\end{proof}

\begin{lem}
Let $\alpha_1,...\alpha_n \geq 0$ with $\sum_i \alpha_i \leq 1$. The map $f:\left(\Reals^+_0\right)^n \rightarrow \Reals^+_0$ given by
\[
f(x_1,..,x_n)=\prod_{i=1}^n x_i^{\alpha_i}
\]
is concave.
\label{lem:concavity}
\end{lem}
\begin{proof}
First assume that all $\alpha_i > 0$ and $\sum_i \alpha_i < 1$.  To show that $f$ is concave we compute the Hesse matrix:
\[
(Hf)_{ij} = \frac{\partial^2}{\partial x_i \partial x_j}f = \frac{\partial}{\partial x_j}\left(\frac{\alpha_i}{x_i}f\right) = \left(-\frac{\alpha_i}{x_i^2}\delta_{ij} + \frac{\alpha_i\alpha_j}{x_ix_j}\right)f =: -A_{ij}f
\]
with $A = \left(\frac{\alpha_i}{x_i^2}\delta_{ij} - \frac{\alpha_i\alpha_j}{x_ix_j}\right)_{1\leq i,j \leq 1}$. To show that $f$ is concave it is sufficient to show that $A$ is positive definite. Let $A^{(k)}$ be the upper left $k\times k$ block of $A$. By Sylvester's criterion $A$ is positive definite if $\det(A^{(k)}) >0 $ for all $k \in \{1,\dotsc,n\}$.

Let $B^{(k)} = \left(\frac{\alpha_i}{x_i^2}\delta_{ij}\right)_{1\leq i,j \leq k}$ and $v^{(k)} = \left(\frac{\alpha_i}{x_i}\right)_{1\leq i \leq k}$. Then
\[
A^{(k)} = B^{(k)} - v^{(k)}(v^{(k)})^T.
\]
By lemma \ref{lem:matrix_determinant_lemma} we calculate
\begin{align}
\det(A^{(k)}) &= \det(B^{(k)})\left(1-(v^{(k)})^T {(B^{(k)})}^{-1} v^{(k)}\right) \nonumber \\
&= \det(B^{(k)})\left(1-\sum_{i=1}^k \left(\frac{\alpha_i}{x_i}\right)^2\frac{x_i^2}{\alpha_i}\right)\nonumber \\
&=\det(B^{(k)})\left(1-\sum_{i=1}^k \alpha_i\right) > 0
\end{align}
where the last inequality follows from $\det(B^{(k)})>0$ and $\sum_{i=1}^k \alpha_i \leq \sum_{i=1}^n \alpha_i <1$. Hence $f$ is concave.

Now assume the general setting where also $\alpha_i=0$ and $\sum_i \alpha_i = 1$ is allowed. For each $i$ choose a sequence $(\alpha^{(m)}_i)_{m \in\Integers}$ such that $\alpha^{(m)}_i >0$, $\sum_i \alpha^{(m)}_i < 1$ and $\alpha^{(m)}_i \xrightarrow{m\rightarrow\infty} \alpha_i$. Let
\[
f^{(m)}(x_1...x_n)=\prod_{i=1}^n x_i^{\alpha^{(m)}_i}.
\]
By continuity
\[
f^{(m)}(x_1...x_n) \xrightarrow{m\rightarrow\infty} f(x_1...x_n).
\]
Now let $x=(x_1,...x_n), y=(y_1...y_n) \in \left(\Reals^+_0\right)^n$ and $\lambda \in [0,1]$. Then
\begin{align}
f(\lambda x + (1-\lambda) y) &= \lim_{m\rightarrow \infty} f^{(m)}(\lambda x + (1-\lambda) y) \nonumber \\
&\geq \lim_{m\rightarrow \infty}  \lambda f^{(m)}(x)+(1-\lambda) f^{(m)}(y) \nonumber \\
&= \lambda f(x)+(1-\lambda) f(y).
\end{align}
\end{proof}

The following lemma is the generalization of lemma \ref{lem:chsh_concavity_argument}.
\begin{lem}
Let $C \subseteq \Reals^d$ be a bounded convex set, and denote by $\vol(C)$ the volume of $C$ (under the Lebesgue measure). Then for any $n \in \Integers$ and any concave function $f:C\rightarrow\Reals^{+}_0$, we have
\[
\int_C f(x)^n \der x \geq \vol(C)\binom{n+d}{d}^{-1} \left(\sup_{x\in C} f(x)\right)^n.
\]
\label{lem:integration}
\end{lem}
The proof idea is similar to that of lemma \ref{lem:chsh_concavity_argument} (assuming for simplicity that $f$ attains its supremum at some point $x^*\in C$): We will lower bound $f$ by a function that is zero on the boundary of $C$, takes the value $f(x^*)$ at $x^*$, and is determined on the rest of $C$ by ``interpolating linearly'' between the values at $x^*$ and the boundary of $C$. (Geometrically, the graph of this new function is basically the surface of a convex cone.)
\begin{proof}
Take any $\epsilon>0$. There exists some $x^*\in C$ such that $f(x^*) \geq \sup_{x\in C} f(x) - \epsilon$. 
We evaluate the integrals using spherical coordinates centered on this point: Let $S^{d-1}\subseteq \Reals^d$ be the $(d-1)$-sphere, and let $\mu$ be the surface measure on $S^{d-1}$ with respect to the usual Lebesgue measure on $\Reals^d$. Since integrals are unchanged by including points on the boundary, we can evaluate the integrals using the closure of $C$ instead, denoted as $\cl{C}$. This is a convex compact set, hence there is a ``radius function'' $R:S^{d-1}\rightarrow \Reals^+_0$ such that
\[\cl{C} = \{x^* + r\Omega | \Omega \in S^{d-1}, r \in [0, R(\Omega)]\}.\]
Then
\[
\vol(C) = \int_{S^{d-1}} \int_0^{R(\Omega)} r^{d-1}\der r \mu(\der\Omega) = \frac{1}{d}\int_{S^{d-1}} R(\Omega)^d \mu(\der\Omega).
\label{eq:CVolume}
\]
Since
\begin{align}
f(x^*+r\Omega) &= f\left(\frac{r}{R(\Omega)}\left(x^*+R(\Omega)\Omega\right)+\frac{R(\Omega)-r}{R(\Omega)}x^*\right) &  \nonumber \\
&\geq \frac{r}{R(\Omega}f(x^*+R(\Omega)\Omega) + \frac{R(\Omega)-r}{R(\Omega)}f(x^*) & \text{Concavity of } f  \nonumber \\
&\geq  \frac{R(\Omega)-r}{R(\Omega)}f(x^*) & \text{Positivity of } f
\label{eq:reduce_f_to_max} 
\end{align}
we have
\begin{align}
\int_C f(x)^n \der x &= \int_{S^{d-1}} \int_0^{R(\Omega)} r^{d-1}f(x^*+r\Omega)\der r \mu(\der\Omega) & \nonumber \\
& \geq f(x^*)^n \int_{S^{d-1}} \int_0^{R(\Omega)} r^{d-1} \left( \frac{R(\Omega)-r}{R(\Omega)}\right)^n \der r \mu(\der\Omega) & \text{by eq.~\eqref{eq:reduce_f_to_max}}  \nonumber \\
&=f(x^*)^n\int_{S^{d-1}} R(\Omega)^d \int_0^1 u^{d-1}(1-u)^n \der u \mu(\der \Omega).
\label{eq:integral_argument}
\end{align}

Using the property of the beta function
\[
\int_0^1 u^{d-1}(1-u)^n = \frac{1}{n+d}\binom{n+d-1}{d-1}^{-1} = \frac{1}{d}\binom{n+d}{d}^{-1}
\]
we find by eq.~\eqref{eq:integral_argument} and eq.~\eqref{eq:CVolume}
\[
\int_C f(x)^n \der x \geq \frac{1}{d}\binom{n+d}{d}^{-1}f(x^*)^n\int_{S^{d-1}} R(\Omega)^d  \mu(\der \Omega) = \binom{n+d}{d}^{-1} \vol(C) f(x^*)^n.
\]
Recalling that $f(x^*) \geq \sup_{x\in C} f(x) - \epsilon$, and $\epsilon>0$ was arbitrary, this implies the desired result.
\end{proof}

When we maximize $Q^{\otimes n}(a|x)$ over $Q \in \Q$ while keeping $a$ and $x$ fixed,  the maximum is achieved for a different $Q$ for each $a$ and $x$. However, the next lemma states that if we average over $Q$ we are at most a polynomial factor below that maximal value of $Q^{\otimes n}(a|x)$, no matter what $a$ and $x$ are.
\begin{lem}
There is a de Finetti state $\tau_{A|X} \in \mathrm{conv}\left(\left\lbrace Q^{\otimes n} | Q \in \Q\right\rbrace\right)$ such that for all \editB{}{$a \in \A$, $x \in \X$}
\[
\tau(a|x) \geq \binom{n+d'}{d'}^{-1} \sup_{Q \in \Q}Q^{\otimes n}(a|x),
\]
where $d'$ is the dimension of the affine hull of $\Q$.
\label{lem:def_tau}
\end{lem}
\begin{proof}
We view $\Q$ as a bounded convex subset of $\Reals^{\editB{}{|\Ah||\Xh|}}$. Because the affine hull of $\Q$ has dimension $d'$ there exists a bounded convex set $C \subseteq \Reals^{d'}$ and a bijective linear map $C \ni \phi \mapsto Q_\phi \in \Q$. Choose
\[
\tau_{A|X} = \frac{1}{\vol(C)} \int_C  Q_\phi ^{\otimes n} \der\phi.
\]
Now fix \editB{}{$a\in\A$ and $x\in\X$} and for \editB{}{$a' \in \Ah$ and $x' \in \Xh$} let $f_{a'x'}=|\{i|a_i = a' \text{ and } x_i=x'\}|/n$ be the frequency of the pair $(a',x')$ in $(a,x)$. 
Then
\[
\phi \mapsto  \left(Q_\phi ^{\otimes n}(a|x)\right)^{1/n} = \prod_{a' \in\A, x'\in\X}  Q_\phi (a'|x')^{f_{a'x'}}
\]
is a concave map by lemma \ref{lem:concavity} and by the linearity of $\phi \mapsto  Q_\phi $. Hence by lemma \ref{lem:integration}
\begin{align}
\tau(a|x) &= \frac{1}{\vol(C)}\int_C Q_\phi ^{\otimes n}(a|x) \der \phi\nonumber \\
& \geq \binom{n+d'}{d'}^{-1}\sup_{\phi\in C} Q_\phi ^{\otimes n}(a|x)\nonumber \\
& =  \binom{n+d'}{d'}^{-1}\sup_{Q\in\Q}Q^{\otimes n}(a|x).
\end{align}
\end{proof}

\subsection{Bounding $P_{A|X}$ by iid boxes}
\label{subsec:bound_pax_by_iid}
In this section we will show that a threshold theorem of the form of eq.~\eqref{eq:threshold_thm} implies that $P(a|x)$ can be bound by $Q^{\otimes n}(a|x)$ up to a polynomial factor. Together with lemma \ref{lem:def_tau} this will yield the proof for theorem \ref{thm:deFinetti_main} in the next section.

\begin{lem}
Let $n=k_1+k_2+...+k_d$ with $k_1,..,k_d \in \Integers$. Then
\[
(n+1)^{-(d-1)} \prod_{r=1}^d \left(\frac{n}{k_r}\right)^{k_r} \leq \binom{n}{k_1,k_2,...,k_d} \leq  \prod_{r=1}^d \left(\frac{n}{k_r}\right)^{k_r}.
\]
\label{lem:multinomial}
\end{lem}
Here $\binom{n}{k_1,k_2,...,k_d} = n!/(k_1!...k_d!)$ is the multinomial coefficient.
\begin{proof}
Let us first prove the inequality for $d=2$, when the multinomial coefficient is just a binomial coefficient. It is a property of the Beta function that
\[
\int_0^1 x^{k_1}(1-x)^{k_2} = \frac{1}{n+1}\binom{n}{k_1}^{-1}
\label{eq:multinomial1}
\]
The integrand on the left hand side is maximized at $x = \frac{k_1}{n}$. Hence
\[
\frac{1}{n+1} \left(\frac{k_1}{n}\right)^{k_1}\left(\frac{k_2}{n}\right)^{k_2} \leq \int_0^1 x^{k_1}(1-x)^{k_2} \leq \left(\frac{k_1}{n}\right)^{k_1}\left(\frac{k_2}{n}\right)^{k_2},
\label{eq:multinomial2}
\]
by lemma \ref{lem:chsh_concavity_argument}. Combining eq.~\eqref{eq:multinomial1} and eq.~\eqref{eq:multinomial2} gives
\[
\frac{1}{n+1}\left(\frac{n}{k_1}\right)^{k_1}\left(\frac{n}{k_2}\right)^{k_2} \leq \binom{n}{k_1} \leq \left(\frac{n}{k_1}\right)^{k_1}\left(\frac{n}{k_2}\right)^{k_2} .
\label{eq:binomial}
\]
Now we prove the lemma for general $d \geq 2$. For this, observe that
\begin{align}
\binom{n}{k_1,...,k_d} &= \binom{n}{k_1}\binom{n-k_1}{k_2}\binom{n-k_1-k_2}{k_3}...\binom{n-k_1-...-k_{d-2}}{k_{d-1}}\nonumber \\
&=\prod_{r=1}^{d-1}\binom{n-\sum_{i=1}^{r-1}k_i}{k_r}.
\end{align}

Applying eq.~\eqref{eq:binomial} to each binomial coefficient completes the proof since
\begin{align}
&\prod_{r=1}^{d-1}\left(\frac{n-\sum_{i=1}^{r-1}k_i}{k_r}\right)^{k_r}\left(\frac{n-\sum_{i=1}^{r-1}k_i}{n-\sum_{i=1}^{r}k_i}\right)^{n-\sum_{i=1}^{r}k_i}\nonumber \\
&=\frac{1}{ \prod_{r=1}^{d-1}k_r^{k_r}}\frac{\prod_{r=1}^{d-1}\left(n-\sum_{i=1}^{r-1}k_i\right)^{n-\sum_{i=1}^{r-1}k_i}}{\prod_{r=1}^{d-1}\left(n-\sum_{i=1}^{r}k_i\right)^{n-\sum_{i=1}^{r}k_i}} \\
&= \prod_{r=1}^d \left(\frac{n}{k_r}\right)^{k_r}.
\end{align}
by a telescoping product argument.
\end{proof}

\begin{lem}
Let $P_{A|X}$ be a n-round box. If 
\[
\Pr_{P_{A|X},\mu^{\otimes n}}\left[ \freq^w(A,X) = f \right] \leq C\exp\left(-\inf_{f' \in \mathcal{F}_{\mu}} D(f||f') n\right)
\]
then
\[
\Pr_{P_{A|X},\mu^{\otimes n}}\left[ \freq^w(A,X) = f \right] \leq C (n+1)^{d-1} \sup_{Q \in \Q} \Pr_{Q^{\otimes n},\mu^{\otimes n}}\left[\freq^w(A,X) = f \right].
\label{eq:bound_by_iid}
\]
\label{lem:pBoundiid}
\end{lem}
\begin{proof}
This is actually just a statement on the two right hand sides. Let $f = (\frac{k_1}{n},...\frac{k_d}{n}) \in \simplex{d}$, let $Q\in\Q$ and let \[f' =  \left(\sum_{\substack{a,x \\ w(a,x)=r}}Q(a|x)\mu(x) \right)_{r=1...d} \in \simplex{d}\] be the element of $\mathcal{F}_\mu$ belonging to $Q$. Then  
\begin{align}
\Pr_{Q^{\otimes n},\mu^{\otimes n}}\left[\freq^w(A,X) = f \right] &= \binom{n}{k_1,...,k_d}\prod_{r=1}^d (f'_r)^{k_r} &  \nonumber \\
&\geq \frac{1}{(n+1)^{d-1}} \prod_{r=1}^d \left(\frac{nf'_r}{k_r}\right)^{k_r} & \text{by lemma \ref{lem:multinomial}} &\nonumber \\
&=\frac{1}{(n+1)^{d-1}} \exp(-D(f||f')n).
\label{eq:iidbound_singleQ}
\end{align}
where the last equality follows directly from the definition of the relative entropy
\[
D(f||f') = \sum_{r=1}^d f_r (\ln(f_r)-\ln(f'_r)).
\]
Taking the supremum over $Q$ gives
\[
\sup_{Q \in \Q} \Pr_{Q^{\otimes n},\mu^{\otimes n}}\left[\freq^w(A,X) = f \right] \geq \frac{1}{(n+1)^{d-1}}\exp\left(-\inf_{f' \in \mathcal{F}_{\mu}} D(f||f') n\right)
\]
which completes the proof.
\end{proof}
\subsection{Proof of theorem \ref{thm:deFinetti_main}}
\label{subsec:proof_main_definetti}
Now we are ready to prove the general theorem \ref{thm:deFinetti_main}. For this, we use lemma \ref{lem:def_tau} to  show that the entries of the de Finetti box are at most smaller by a polynomial factor then the corresponding entries of any iid box. Then we use lemma \ref{lem:pBoundiid} to show that the threshold theorem implies that the probability of a frequency $f$ under $P_{A|X}$ can be bounded, up to a polynomial factor, by probability of $f$ under some iid box (but possibly a different iid box for each $f$). Combining both lemmas yields the proof of part 1. Part 2 will follow from part 1 by using the definition of $w$-symmetry, and part 3 will follow directly from lemma \ref{lem:multinomial}. 

\begin{proof}[Proof of theorem \ref{thm:deFinetti_main}]
\begin{enumerate}
\item This part follows directly by combining lemma \ref{lem:pBoundiid} and  lemma \ref{lem:def_tau}. Suppose
\[
\Pr_{P_{A|X},\mu^{\otimes n}}\left[ \freq^w(A,X) = f \right] \leq C\exp\left(-\inf_{f' \in \mathcal{F}_{\mu}} D(f||f') n\right).
\]
By lemma \ref{lem:pBoundiid} we have
\[
\Pr_{P_{A|X},\mu^{\otimes n}}\left[ \freq^w(A,X) = f \right] \leq C (n+1)^{d-1} \sup_{Q \in \Q} \Pr_{Q^{\otimes n},\mu^{\otimes n}}\left[\freq^w(A,X) = f \right].
\]
Since by lemma \ref{lem:def_tau}
\[
\Pr_{Q^{\otimes n},\mu^{\otimes n}}\left[\freq^w(A,X) = f \right] \leq \binom{n+d'}{d'}\Pr_{\tau,\mu^{\otimes n}}\left[\freq^w(A,X) = f \right] 
\]
it follows
\[
\Pr_{P_{A|X},\mu^{\otimes n}}\left[ \freq^w(A,X)=f\right] \leq C \binom{n+d'}{d'} (n+1)^{d-1} \Pr_{\tau,\mu^{\otimes n}}\left[ \freq^w(A,X)=f\right].
\]

\item 
For this part, we have by hypothesis that $P_{A|X}$ and every box in $\Q$ has $w$-symmetry. Then also $\tau_{A|X}$ has $w$-symmetry. Take any \editB{}{$a\in\A, x\in \X$}, and define $f=\freq^w(a|x)$. Then we have 
\[
\Pr_{P_{A|X},\mu^{\otimes n}}\left[ \freq^w(A,X)=f\right] = P(a|x)\sum_{\substack{a\in\A^n, x\in\X^n \\ \freq^w(a,x)=f}}\mu^{\otimes n}(x),
\label{eq:w-symmetry1}
\]
and similarly 
\[
\Pr_{\tau_{A|X},\mu^{\otimes n}}\left[ \freq^w(A,X)=f\right] = \tau(a|x) \sum_{\substack{a\in\A^n, x\in\X^n \\ \freq^w(a,x)=f}}\mu^{\otimes n}(x).
\label{eq:w-symmetry2}
\]
From eq.~\eqref{eq:w-symmetry1}, eq.~\eqref{eq:w-symmetry2} and part 1 it follows that
\[
P(a|x) \leq C \binom{n+d'}{d'} (n+1)^{d-1} \tau(a|x).
\]

\item Now assume
\[
P(a|x) \leq \tilde{C} \tau(a|x)
\]
with $\tilde{C} = C\binom{n+d'}{d'} (n+1)^{d-1}$. 
Since $\tau(a|x) \leq \sup_{Q \in \Q} Q^{\otimes n}(a|x)$ it follows for $f=\left(\frac{k_1}{n}...\frac{k_d}{n}\right)$ that
\begin{align}
\Pr_{P_{A|X},\mu^{\otimes n}}\left[ \freq^w(A,X) = f \right] &\leq \tilde{C} \sup_{Q \in \Q} \Pr_{Q^{\otimes n},\mu^{\otimes n}}\left[\freq^w(A,X) = f \right]  & \nonumber \\
&= \tilde{C}\sup_{f' \in \mathcal{F}_\mu} \binom{n}{k_1,...,k_d} \prod_{r=1}^d (f'_r)^{k_r}& \nonumber \\
&\leq \tilde{C}\sup_{f' \in \mathcal{F}_\mu} \prod_{r=1}^d \left(\frac{nf'_r}{k_r}\right)^{k_r} & \text{by lemma \ref{lem:multinomial}} \nonumber \\
&= \tilde{C}\exp\left(-\inf_{f' \in \mathcal{F}_{\mu}} D(f||f') n\right).&
\end{align}
\end{enumerate}
\end{proof}

\bibliographystyle{unsrtnat}
\bibliography{bibliography}
\end{document}